 \documentclass[runningheads]{llncs}  

  \usepackage[margin=1in]{geometry}
\usepackage{blindtext}
\usepackage{hyperref}
\usepackage{graphicx}
\usepackage{amsfonts} 
\usepackage{booktabs} 
\usepackage{algpseudocode}
\usepackage{algorithm}
\usepackage[flushleft]{threeparttable}
\usepackage{pifont}
\algdef{SE}{Begin}{EndBegin}{\textbf{begin}}{\textbf{end}}

\newcommand{\IR}{\mathbb{R}}

\def\OO{{\cal O}} 

\def\A{{\cal A}}

\def\I{{\mathcal I}}

\def\P1{{\mathcal P_1}}
\def\P2{{\mathcal P_2}}
\def\P{{\mathcal P}}

\def\C{{\mathcal C}}
\newcommand{\AFF}{\textsf{\sc First-Fit}}
\newcommand{\AGDS}{\textsf{\sc GreedyDS}}
\newcommand{\AL}{\textsf{\sc Layer}}
\newcommand{\AGCDS}{\textsf{\sc GreedyCDS}}

 \usepackage{mathtools}
 \usepackage{pdfpages}
\usepackage{caption}
\usepackage{subcaption}
 \newtheorem{observation}{Observation}
\usepackage{color}

\newcommand{\myqed}{\hfill $\Box$}
\begin{document}
%


   \title{Online Dominating Set and Coloring for Geometric Intersection Graphs\thanks{Preliminary version of this paper will appear in the proceedings of the 16th annual international conference on combinatorial optimization and applications (COCOA), 2023~\cite{DeKS23}. Work on this paper by M. De has been partially supported 
      by  SERB MATRICS Grant MTR/2021/000584, and 
      work by S. Singh has been supported 
      by CSIR (File Number-09/086(1429)/2019-EMR-I).}}

%
%

\author{Minati De\inst{1}\orcidID{0000-0002-1859-8800} \and Sambhav Khurana\inst{2} \and
Satyam Singh\inst{1}}
\authorrunning{M. De, S.Khurana and S. Singh}
%
\institute{Dept. of Mathematics, Indian Institute of Technology Delhi, New Delhi, India
\email{\{minati,satyam.singh\}@maths.iitd.ac.in}\\
\and
Dept. of Computer Science \& Engineering, Texas A\&M University, USA\\
\email{\ sambhav{\textunderscore}khurana@tamu.edu}}

\maketitle              

\begin{abstract}
We present online deterministic algorithms for minimum coloring and minimum dominating set problems in the context of geometric intersection graphs. We consider a graph parameter: the independent kissing number $\zeta$,  which is a number equal to `the size of the largest induced star in the graph $-1$'. For a graph with an independent kissing number at most $\zeta$, we show that the famous greedy algorithm achieves an optimal competitive ratio of $\zeta$ for the minimum dominating set and the minimum independent dominating set problems. However, for the minimum connected dominating set problem, we obtain a competitive ratio of at most $2\zeta$. To complement this, we prove that for the minimum connected dominating set problem, any deterministic online algorithm has a competitive ratio of at least $2(\zeta-1)$ for the geometric intersection graph of translates of a convex object in $\IR^2$. Next, for the minimum coloring problem, we obtain algorithms having a competitive ratio of $O\left({\zeta'}{\log m}\right)$ for geometric intersection graphs of bounded scaled $\alpha$-fat objects in $\IR^d$ having widths in the interval $[1,m]$, where $\zeta'$ is the independent kissing number of the geometric intersection graph of bounded scaled $\alpha$-fat objects having widths in the interval $[1,2]$.
Finally, we investigate the value of $\zeta$ for geometric intersection graphs of various families of geometric objects.

\end{abstract}

\section{Introduction}
\pagenumbering{arabic}
  We consider deterministic online algorithms for some well-known NP-hard problems: the minimum dominating set problem and its variants and the minimum coloring problem. While dominating set and its variants have several applications in wireless ad-hoc networks, routing, etc.~\cite{BoyarEFKL19,BoyarFKL22}, coloring has diverse applications in frequency assignment, scheduling and many more~\cite{AlbersS21,ErlebachF02,ErlebachF06}.
  \paragraph{Minimum Dominating Set and its Variants}
 For a graph $G=(V,E)$, a subset $D \subseteq V$  is a \emph{{dominating set}} (DS) if for each vertex $v \in V$, either $v\in D$ (containment) or there exists an edge $\{u,v\}\in E$ such that $u\in D$ (dominance).  A dominating set $D$ is said to be a \emph{{connected dominating set}} (CDS) if the induced subgraph $G[D]$ is connected (if $G$ is not connected, then $G[D]$ must be connected for each connected component of $G$). A dominating set $D$ is said to be an \emph{{independent dominating set}} (IDS) if the induced subgraph $G[D]$ is an independent set.  The \emph{{minimum dominating set}} (MDS) problem involves finding a dominating set of the minimum cardinality. Similarly, the objectives of the \emph{{minimum connected dominating set}} (MCDS) problem and the \emph{{minimum independent dominating set}} (MIDS) problem are to find a CDS and IDS, respectively, with the minimum number of vertices.


Throughout the paper, we consider online algorithms for \emph{vertex arrival model} of graphs where a new vertex is revealed with its edges incident to previously appeared vertices.
The dominating set and its variants can be considered in various online models~\cite{BoyarFKL22}. 
In \emph{Classical-Online-Model} (also known as ``Standard-Model''), upon the arrival of a new vertex, an online algorithm must either accept the vertex by adding it to the solution set or reject it. 
In \emph{Relaxed-Online-Model} (also known as ``Late-Accept-Model''), 
upon the arrival of a new vertex, in addition to the revealed vertex, an online algorithm may also include any of the previously arrived vertices to the solution set. Note that once a vertex is included in the solution set for either model, the decision cannot be reversed in the future.
For the MCDS problem, if we cannot add previously arrived vertices in the solution set, the solution may result in a disconnected dominating set~\cite{BoyarEFKL19}. Therefore, in this paper, we use Classical-Online-Model for MDS and MIDS problems, while for the case of the MCDS problem, we use Relaxed-Online-Model. Furthermore, the revealed induced subgraph must always be connected for the MCDS problem.

 \paragraph{Minimum Coloring Problem} For a graph $G=(V,E)$, the \emph{coloring} is to assign colors (positive integers) to the vertices of $G$. The \emph{minimum coloring} (MC) problem is to find a coloring with the minimum number of distinct colors such that no two adjacent vertices (vertices connected by an edge) have the same color.
In the online version, upon the arrival of each new vertex $v$, an algorithm needs to immediately assign a feasible color to $v$, i.e.,  one distinct from the colors assigned to the neighbours of $v$ that have been already arrived. The color of the vertex $v$ cannot be changed in future.


The quality of our online algorithm is analyzed using competitive analysis~\cite{BorodinE}.
An online algorithm $ALG$ for a minimization problem is said to be \emph{$c$-competitive}, if there exists a constant $d$ such that for any input sequence $\cal I$, we have $\A({\cal I}) \leq c\times\OO({\cal I}) +d$, where $\A({\cal I})$ and $\OO({\cal I})$ are the cost of solutions produced by  $ALG$ and  an optimal offline algorithm, respectively, for the input sequence $\cal I$. The smallest $c$ for which $ALG$ is $c$-competitive is known as an \emph{{asymptotic competitive ratio}} of $ALG$. The smallest $c$ for which $ALG$ is $c$-competitive with  $d=0$ is called an \emph{{absolute competitive ratio}} (also known as ``strict competitive ratio") of $ALG$. If not explicitly specified, we use the term ``competitive ratio'' to mean absolute  competitive ratio. 

In this paper, we focus on geometric intersection graphs due to their applications in wireless sensors, network routing, medical imaging, etc~\cite{AgarwalKS98,BachmannHS13,BakshiCW20,ButenkoPSSS02,HalldorssonIMT02,Malesinska1997}. For a family ${\cal S}$ of geometric objects in $\IR^d$, the \emph{geometric intersection graph} $G$ of ${\cal S}$ is an undirected graph with set of vertices same as ${\cal S}$, and the set of edges is defined as $E=\{\{u,v\}| u,v\in {\cal S}, u\cap v\neq \emptyset \}$.
Several researchers have used the kissing number (for definition, see Section~\ref{Notations}) as a parameter to give an upper or lower bound for geometric problems. For instance, Butenko et al.~\cite{ButenkoKU11} used it to prove the upper bound of the MCDS problem in the offline setup for unit balls in $\IR^3$; whereas Dumitrescu et al.~\cite{DumitrescuGT20} used it to prove the upper bound for the unit covering problem for unit balls in $\IR^d$.
Similar to kissing number, we use a graph parameter- independent kissing number $\zeta$.
Let $\varphi(G)$  denote the size of a maximum independent set of a graph $G=(V,E)$.  For any vertex $v\in V$,  let $N(v)=\{u(\neq v) \in V | \{u,v\} \in E\}$ be the neighbourhood of the vertex $v$.
  Now, we define the \emph{{independent kissing number}} $\zeta$ for graphs. 
\begin{definition}[Independent Kissing Number]\label{def1}
 The independent kissing number $\zeta$ of a graph  $G=(V,E)$  is defined as $\max_{v \in V} \{\varphi(G[N(v)])\}$.
\end{definition}
Note that the independent kissing number equals `the size of the largest induced star in the graph $-1$'. In other words, a graph with independent kissing number $\zeta$ is a $K_{1,\zeta+1}$-free graph (for definition, see Section~\ref{sect:graph}). Moreover,  the value of $\zeta$ may be very small compared to the number of vertices in a graph. For example, the value of~$\zeta$ is a fixed constant for the geometric intersection graph of several families of geometric objects like translated and rotated copies of a convex object in $\IR^2$. 
However, the use of this parameter is not new. For example, in the offline setup, Marathe et al.~\cite{MaratheBHRR95} obtained a $2(\zeta -1)$-approximation algorithm for the  MCDS problem for any graph having an independent kissing number at most $\zeta$.




   


\subsection{Related Work}

The dominating set and its variants are well-studied in the offline setup.
Finding MDS is known to be NP-hard even for unit disk graphs~\cite{ClarkCJ90,10.6/574848}. A polynomial-time approximation scheme (PTAS) is known when all objects are homothets {(bounded and scaled copies)} of a convex object~\cite{DeL23}.
King and Tzeng~\cite{KingT97} initiated the study of the online MDS problem in Classical-Online-Model. They showed that for a general graph, the deterministic greedy algorithm achieves a competitive ratio of $n-1$, which is also a tight bound achievable by any deterministic online algorithm for the MDS problem, where~$n$ is the length of the input sequence (i.e., number of vertices appeared). Even for the interval graph, the lower bound of the competitive ratio is~$n-1$~\cite{KingT97}.  Eidenbenz~\cite{Eidenbenz} proved that the greedy algorithm achieves a tight bound of~5 for finding MDS of the unit disk graph. 

Boyar et al.~\cite{BoyarEFKL19} considered a variant of  the Relaxed-Online-Model for MDS, MIDS and MCDS problems in which, in addition to the  Relaxed-Online-Model the revealed graph should always be connected.
In the aforementioned setup, they studied MDS, MIDS and MCDS problems for specific graph classes such as trees, bipartite graphs, bounded degree graphs, and planar graphs. Their results are summarized in~\cite[ Table 2]{BoyarEFKL19}. They also proposed a~3-competitive deterministic algorithm for the MDS problem in the above-mentioned model for a tree. Later, Kobayashi~\cite{Kobayashi17} proved that~3 is also the lower bound for trees in this setting.
In the same setup, Eidenbenz~\cite{Eidenbenz}  showed that, for the MCDS problem of unit disk graphs, the deterministic greedy algorithm achieves a competitive ratio of $8+\epsilon$. In contrast, no deterministic online algorithm can guarantee a  strictly better competitive ratio than $10/3$.
We observe that the (asymptotic) competitive ratio of the MCDS problem for unit disk graphs could be improved to~6.798 (see Section~\ref{2.2}).

The minimum coloring problem is a celebrated problem in combinatorial optimization~\cite{KargerMS98,MATULA1972109,Wigderson83}. The problem is known to be NP-hard,  even for unit disk graphs~\cite{ClarkCJ90}. In the offline setup, a 3-approximation algorithm for unit disk graphs was presented by Peeters~\cite{Peeters}, Gr$\ddot{\text{a}}$f et al.~\cite{GrafSW98} and Marathe et al.~\cite{MaratheBHRR95}, independently. Generalizing this, Marathe et al.~\cite{MaratheBHRR95}  obtained a 6-approximation algorithm for disk graphs. Later Erlebach and Fiala~\cite{ErlebachF06} sharpened the analysis and showed that the algorithm is, in fact,  5-approximation. Surprisingly, there is a lack of study of this problem in the offline setup for other geometric intersection graphs.


For the online minimum coloring problem, the algorithm $\AFF$ is one of the most popular algorithms. For general graphs with $n$ vertices, while the competitive ratio of  $\AFF$ is $n/4$, Lov{\'{a}}sz et al.~\cite{LovaszST89} 
presented an algorithm that achieves a sublinear competitive ratio.
In the context of geometric intersection graphs, Erlebach and Fiala~\cite{ErlebachF02} proved that the algorithm $\AFF$ achieves a competitive ratio of $O(\log n)$ for disk graphs and square graphs. They also showed that no deterministic online
algorithm can achieve a competitive ratio better than $\Omega(\log n)$ for disk graphs and for square graphs,
even if the geometric representation is given as part of the input. 
Recently, Albers and Schraink~\cite{AlbersS21} proved that $\Omega(\log n)$ is also the lower bound of the competitive ratio
of any randomized online algorithms for the MC problem for disk graphs. As a result, $\AFF$ is an asymptotically optimal competitive algorithm for disk graphs.
Capponi and Pilloto~\cite{CapponiP05} proved that for any graph with an independent kissing number at most $\zeta$, the popular algorithm $\AFF$ achieves a competitive ratio of at most $\zeta$. 
Erlebach and Fiala~\cite{ErlebachF02} presented a deterministic algorithm $\AL$ having a competitive ratio of  $O(\log m)$ for bounded scaled disks having radii in the interval $[1,m]$. The algorithm $\AL$ uses the disk representation. Later, Caragiannis at al.~\cite{CaragiannisFKP07b} proved that $\AFF$, which does not use disk representation, is also $O(\log m)$ competitive for this case.   In this paper, we generalize this result for the geometric intersection graph of bounded scaled $\alpha$-fat objects in $\IR^d$.

\subsection{Our Contributions}

In this paper, we obtain the following results.

\begin{enumerate}
   
\item First, we prove that for MDS and MIDS problems, a well-known algorithm $\AGDS$ has a competitive ratio of at most $\zeta$ for a graph with an independent kissing number at most $\zeta$. This result is tight since the competitive ratios of any deterministic online algorithm for these problems is at least $\zeta$ (Theorem~\ref{th:mds_geometric_objects_1} and Theorem~\ref{th:mids_geometric_objects}). 

 \item  For the MCDS problem, we prove that, for any graph having the independent kissing number at most~$\zeta$, an algorithm $\AGCDS$ achieves a competitive ratio of at most $2\zeta$ (Theorem~\ref{cds}). To complement this,  we prove that the lower bound of the competitive ratio is at least~${2(\zeta-1)}$ which holds even for a  geometric intersection graph of translates of a convex object in $\IR^2$   (Theorem~\ref{thm:mcds_con}).
 
 \item  Next, we consider the MC problem for geometric intersection graphs of bounded scaled $\alpha$-fat objects in $\IR^d$ having widths in the interval $[1,m]$. For this, due to \cite{CapponiP05}, the best known competitive ratio is $\zeta$, where $\zeta$  is the independent kissing number of bounded scaled $\alpha$-fat objects having widths in the interval $[1,m]$.
Inspired by Erlebach and Fiala~\cite{ErlebachF02}, we present an algorithm $\AL$  having a competitive ratio of at most $O\left({\zeta'}{\log m}\right)$, where  $\zeta'$ is the independent kissing number of bounded scaled $\alpha$-fat objects having widths in the interval $[1,2]$ (Theorem~\ref{thm:AL}). Since the value of $\zeta$ could be very large compared to ${\zeta'}{\log m}$ (see Remark~\ref{zeta_zeta'}), it is a significant improvement. Note that this algorithm uses the geometric representation of the objects. Next, we prove that the algorithm $\AFF$, which does not require object representation,  achieves asymptotically the same competitive ratio.


 \item All above-obtained results for the MC problem, MDS problem and its variants depend on the graph parameter: the independent kissing number $\zeta$. Therefore, the value of~$\zeta$  becomes a crucial graph parameter to investigate. To estimate the value of~$\zeta$, we consider various families of geometric objects.
We show that for congruent balls in~$\IR^3$  the value of~$\zeta$ is 12. For translates of a regular $k$-gon ($k\in([5,\infty)\cup\{3\})\cap\mathbb{Z}$) in~$\IR^2$, we show that $5\leq \zeta\leq 6$. {For translates of a hypercube and congruent hypercubes in~$\IR^d$, the value of $\zeta$ is $2^d$ and $2^{d+1}$, respectively. We also give bounds on the value of~$\zeta$ for $\alpha$-fat objects in $\IR^d$ having widths in the interval $[1,m]$ and disks in $\IR^2$ having radii in the interval $[1,2]$.} We feel that these results will find applications in many problems. We illustrate a few in Section~\ref{Application_zeta}. It is also summarized in Table~\ref{tab:competitive_ratios}.
\end{enumerate}

Note that all of our algorithms are deterministic, and they need not know the value of $\zeta$  in advance. In particular, algorithms $\AGDS$, $\AGCDS$ and $\AFF$ do not need to know the object's representation; whereas, the algorithm $\AL$ needs to know the object's width upon its arrival.

\begin{table}[htbp]
\centering
\caption{Obtained Competitive Ratios for Different Problems}
\label{tab:competitive_ratios}
\begin{threeparttable}
\begin{tabular}{p{4.75cm}p{1.75cm}p{3.25cm}p{3.25cm}p{2 cm}}
\toprule
\textbf{Geometric Intersections Graphs of} &\textbf{MDS, MIDS}& \textbf{MC} & \textbf{MCDS} & \textbf{t-Relaxed Coloring} \\
\midrule
  Congruent balls in $\mathbb{R}^3$ & $12$\tnote{*} & 12 &22 &288 \\
 \midrule
 Translated copies of a hypercube in $\mathbb{R}^d$, where $d \in \mathbb{Z}^+$ & $2^{d}$\tnote{*} & $2^d$ & $2(2^d-1)$ & $2^{2d+1}$\\
 \midrule
 Congruent hypercubes in $\IR^d$, where $d \in \mathbb{Z}^+$ & $2^{d+1}$&$2^{d+1}$ & $2(2^{d+1}-1)$& $2^{2d+3}$\\
 \midrule
 Translated copies of a regular $k$-gon (for $k = 3$ and $k \geq 5$) &  6 & 6& 10 & 72 \\
\midrule
  Bounded scaled disks in $\mathbb{R}^2$ with radii in $[1, 2]$ & $11$\tnote{*}  & 11&20 &242 \\
 \midrule
  Bounded-scaled $\alpha$-fat objects in $\mathbb{R}^d$ with widths in $[1, m]$, where $d \in \mathbb{Z}^+$ &    $\Big(\frac{m}{\alpha}+2\Big)^d$ & $O\Big(\Big(\frac{2}{\alpha}+2\Big)^d\log m\Big)$ & $2\left(\Big(\frac{m}{\alpha}+2\Big)^d-1\right)$ & $2\Big(\frac{m}{\alpha}+2\Big)^{2d}$\\
  \bottomrule
  \vspace{-0.75 cm}
\end{tabular}
\begin{tablenotes}
\item[*] The result is optimal due to matching the lower bound.
\end{tablenotes}
\end{threeparttable}
\end{table}

\subsection{Organization}
First, we give some relevant definitions and preliminaries in 
 Section~\ref{Notations}.
In Section~\ref{MDS_variants}, we discuss the performance of well-known online greedy algorithms for MDS and its variants. Then in Section~\ref{sec:lb_MCDS}, we propose a lower bound of the MCDS problem for translated copies of a convex object in $\IR^2$. In Section~\ref{sec:coloring}, we analyse the performance of the algorithm $\AL$ and $\AFF$ for a family of bounded scaled $\alpha$-fat objects in $\IR^d$ having widths in the interval $[1,m]$.
After that, in Sections~\ref{sec:IKS}, we study the value of $\zeta$  for specific families of geometric intersection graphs, and we illustrate a few applications in Section~\ref{Application_zeta}.
Finally, in Section~\ref{10}, we give a conclusion.

\section{Notation and Preliminaries}\label{Notations}
We use $\mathbb{Z}^{+}$ to denote the set of positive integers. For $n\in\mathbb{Z}^{+}$, we use $[n]$ to denote the set $\{1,2,\ldots,n\}$. We use $\mathbb{Z}_n$ to denote the set $\{0,1,\ldots,n-1\}$ of non-negative integers less than $n$, where the arithmetic operations such as addition, subtraction, and multiplication are modulo $n$.

\subsection{Related to Geometric Objects}
By a \emph{geometric object}, we refer to a compact (i.e., closed and bounded) convex set in $\mathbb{R}^d$ with a nonempty interior.
The \emph{center} of a geometric object is defined as the center of the smallest radius ball inscribing that geometric object. 
We use  $\sigma(c)$ to represent a geometric object $\sigma$ centered at a point $c$. {We use the term \emph{unit hypercube} to denote a hypercube having side length one; while we use \emph{unit ball} to denote a ball having radius one.}
Two objects are \emph{{congruent}} if one can be transformed into the other by a combination of rigid motions, namely translation, rotation, and reflection.
Two geometric objects are said to be \emph{{non-overlapping}} if they have no common interior; whereas we call them \emph{{non-touching}} if their intersection is empty. Note that two non-overlapping objects may intersect in their boundary.  

 \subsection{Related to Graph}\label{sect:graph} 
For any subset $V'\subseteq V$ of a graph $G=(V,E)$, the \emph{induced subgraph} $G[V']$ is a graph whose vertex set is $V'$ and whose edge set consists of all of the edges in $E$ that have both endpoints in  $V'$. For a given graph $H$, a graph $G$ is \emph{$H$-free} if $G$ does not contain $H$ as an induced subgraph.
A graph $G=(V,E)$ of order at least 2 is said to be \emph{bipartite} if its vertex set $V$ can be partitioned into two non-empty subsets $V_1$ and $V_2$ such that each of the edges in $G$ connects a vertex in $V_1$ to a vertex in $V_2$. The sets $V_1$ and $V_2$ are known as the partite sets.
A \emph{complete bipartite} graph is a bipartite graph such that two vertices are adjacent if and only if they are in different partite sets. When the partite sets $V_1$ and $V_2$ have sizes $m$ and $n$, respectively, then the corresponding complete bipartite graph is denoted as $K_{m,n}$.

\subsection{Kissing Number vs Independent Kissing Number}\label{K_vs_zeta}

We give a graph theoretic definition of {{independent kissing  number}} in Definition~\ref{def1}.
 {An alternative definition  for a family $\cal S$ of geometric objects in $\mathbb{R}^d$ is as follows.}
\begin{definition}\label{def2}
Let $\cal S$ be a family of geometric objects, and 
let $u$ be any object belonging to the family ${\cal S}$.  Let $\zeta_{u}$ be the maximum number of pairwise non-touching objects in $\cal S$ that we can arrange in and around $u$ such that all of them are intersected by $u$. The independent kissing number  $\zeta$ of  $\cal S$ is defined to be $\max_{u\in {\cal S}}\zeta_u$.
\end{definition}

 A set $K$ of objects belonging to the family $\cal S$ is said to form an \emph{independent kissing configuration} if (i) there exists an object $u\in K$  that intersects all objects in $K\setminus\{u\}$, and (ii) all objects in   $K\setminus\{u\}$ are mutually non-touching to each other. Here $u$ and $K\setminus\{u\}$ are said to be the \emph{core} and \emph{independent set}, respectively, of the independent kissing configuration. The  configuration is  considered \emph{optimal} if $| K\setminus\{u\}| =\zeta$, where $\zeta$ is the independent kissing number of $\cal S$.
The configuration is said to be \emph{standard} if all objects in $K\setminus\{u\}$ are mutually non-overlapping with $u$, i.e., their common interior is empty but touches the boundary of $u$.

For a convex object $u \subset \mathbb{R}^d$, the \emph{{Kissing number}} (also known as \emph{{Newton number}}), denoted as $\kappa(u)$, is the maximum number of non-overlapping congruent copies of $u$ that can be arranged around $u$ so that each of them touches $u$~\cite{BrassMP}.
\begin{figure}[htbp]
     \begin{subfigure}[b]{0.33\textwidth}
         \centering
\includegraphics[width=32 mm]{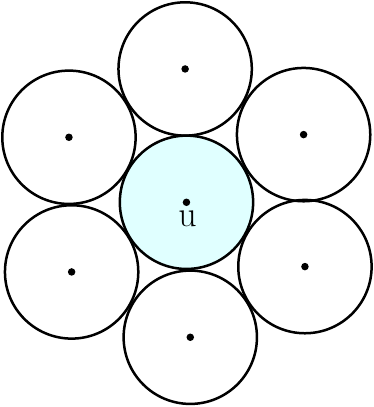} 
\caption{}
\label{fig:kissing_number}
     \end{subfigure}
     \hfill
     \begin{subfigure}[b]{0.32\textwidth}
         \centering
\includegraphics[width=33 mm]{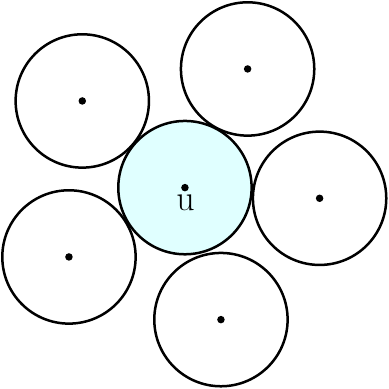}
\caption{}
\label{fig:Independent kissing}
     \end{subfigure}
 \hfill
     \begin{subfigure}[b]{0.33\textwidth}
        \centering
\includegraphics[width=36 mm]{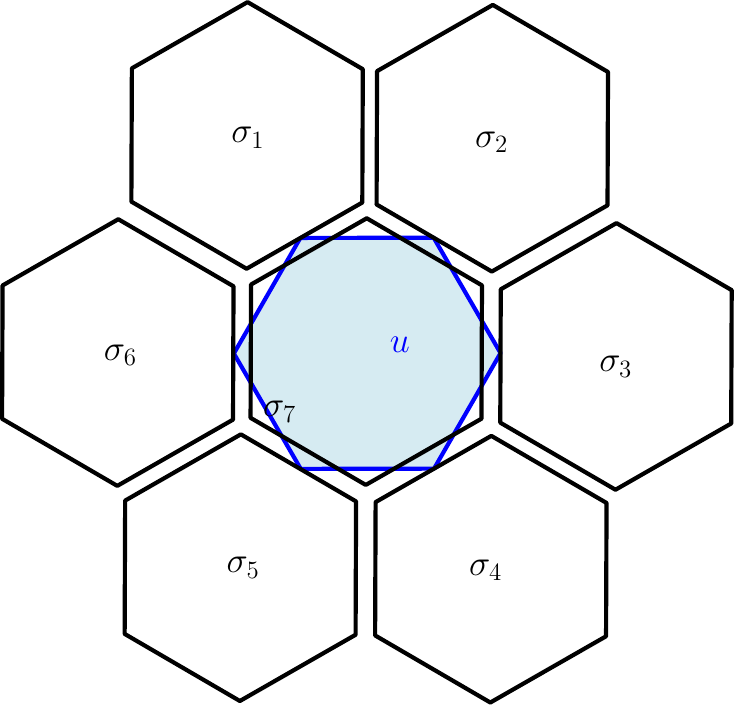}
\caption{}
\label{fig:cong_ind}
     \end{subfigure}
       \caption{(a) Optimal kissing configurations for unit disks; (b)  Optimal independent kissing configurations for unit disks; (c) An optimal independent kissing configuration for congruent hexagons.}
       \label{fig:configuration}
\end{figure}
Note that the kissing number is defined mainly for a set of congruent objects. On the other hand, we define the independent kissing number for any set of geometric objects.
In the kissing configuration, objects around $u$ are non-overlapping but can touch each other; on the other hand, in the independent kissing configuration, objects around $u$ are non-touching. For more illustration, see Figure~\ref{fig:configuration}. 
For a set of congruent copies (respectively, translated copies) of $u$, it is easy to observe that the value of the independent kissing number is at most $\kappa(u)+1$ (respectively, $\kappa(u)$), where $\kappa(u)$ is the kissing number of $u$.  For example, Zhao~\cite{ZHAO1998293} proved that for the hexagon, the value of the kissing number is 6. On the other hand, it is easy to observe from Figure~\ref{fig:cong_ind} that the value of the independent kissing number is 7 for a family of congruent hexagons.

\subsection{Fat Objects}\label{App_fat}
 A number of different definitions of fatness (not extremely long and skinny) are available in the geometry literature~\cite{BergSVK02,Chan03,EfratKNS00}. For our purpose, we define the following.
Let $\sigma$ be an {object} and $x$ be any point in $\sigma$. Let {$\alpha(x)$} be the ratio between the minimum and the maximum distance (under Euclidean norm) from $x$ to the boundary $\partial(\sigma)$ of the object $\sigma$. In other words, $\alpha(x)= \frac{\min_{y\in \partial(\sigma)}d(x,y)}{\max_{y\in \partial(\sigma)}d(x,y)}$, where $d(.,.)$ denotes the Euclidean distance.
	The \emph{aspect ratio} $\alpha$ of an object $\sigma$ is defined as the maximum value of $\alpha(x)$ for any point  $x\in \sigma$, i.e., $\alpha=\max \{\alpha(x)  : x \in  \sigma\}$. An object is said to be \emph{$\alpha$-fat object} if its aspect ratio is $\alpha$. Observe that $\alpha$-fat objects are invariant under translation, rotation and scaling.
 The \emph{aspect point} of  $\sigma$ is a point in $\sigma$ where the aspect ratio of $\sigma$ is attained.
 The minimum distance (respectively, maximum distance) from  the aspect point to the boundary of the object is referred to as the \emph{width} (respectively, \emph{height}) of the object. Observe that $\alpha$-fat objects are invariant under translation, rotation and scaling.

\subsection{Convex Distance Function}
Let $x,y$ be any two points in $\IR^d$.
Let $C$ be a convex object containing the origin in its interior.
 We translate $C$ by a vector $x$ and consider the ray from $x$ through $y$. Let $v$ denote
the unique point on the boundary of $C$ hit by this ray.
The convex distance function $d_{C}$, induced by the object $C$,
is defined as $d_{C}(x,y)=\frac{d(x,y)}{d(x,v)}$~\cite{IckingKLM95}.


\section{Dominating Set and its Variants}\label{MDS_variants}

In this section, we discuss well-known greedy online algorithms ($\AGDS$ and $\AGCDS$) for MDS and MCDS problems for graphs and show how their performance depends on the parameter, independent kissing number, $\zeta$. We would like to mention that all the algorithms presented in this section need not know the value of $\zeta$  in advance, and the object's representation is also unnecessary.

\subsection{Minimum Dominating Set}\label{2.1}
The algorithm $\AGDS$ for 
 finding a minimum dominating set is as follows.
The algorithm maintains a feasible dominating set $\A$. Initially, $\A=\emptyset$. On receiving a new vertex $v$, if the vertex is not dominated by the existing dominating set $\A$, then update $\A\leftarrow\A \cup \{v\}$.
Eidenbenz~\cite{Eidenbenz}  showed that this algorithm achieves an optimal competitive ratio of $5$ for the unit disk graph. It is easy to generalize this result for graphs with a fixed independent kissing number $\zeta$.

\begin{observation}\label{obs1}
The vertices returned by the algorithm $\AGDS$ are
pairwise non-adjacent. In other words, the solution set is always an independent set. 
\end{observation}

 \begin{theorem}\label{th:mds_geometric_objects_1}
 {The algorithm $\AGDS$ has a competitive ratio of at most $\zeta$ for the MDS problem of a graph  having an independent kissing number at most~$\zeta$. This result is tight: the competitive ratio of any online algorithm for this problem is at least $\zeta$.}
\end{theorem}
\begin{proof}
    
\textbf{Upper Bound Obtained by $\AGDS$.}
Let $G=(V,E)$ be a graph received by the algorithm. Let $\zeta$ be the independent kissing number of $G$.
 Let $\A$ be the set of vertices reported by $\AGDS$. 
Similarly, let $\OO$ be the set of vertices reported by an offline optimal algorithm for $G$.
Let $\C=\A\cap \OO$. Let $\A'=\A \setminus \C$ and $\OO'=\OO\setminus \C$. 

Let $o\in \OO$ be any vertex in the optimum solution. Let $\A_o\subseteq \A'$ be the set of vertices in $\A'$ that are adjacent to the vertex $o$. Since the graph $G$  has an independent kissing number $\zeta$ and the  vertices in $\A_o$ are pairwise non-adjacent (by Observation~\ref{obs1}), we have $|\A_o|\leq~\zeta$.

Let $a$ be any vertex in $\A'$. Since $\OO$ also dominates $\A'$ and  $a\notin \OO$, the vertex $a$ must be  a neighbour of some vertex in $\OO$. 
Since $\A$ is also an independent set (by Observation~\ref{obs1}), the vertex $a$ cannot be a neighbour of some vertex in $\C$.
As a result, the vertex $a$ must be a neighbour of some vertex in $\OO'$.  In other words, $a\in \A_o$ for some $o\in \OO'$. Therefore, $\cup_{o\in \OO'}\A_o=\A'$. Hence, $|\A'|\leq \zeta|\OO'|$, resulting in $|\A|\leq \zeta|\OO|$. Thus, our algorithm achieves a competitive ratio of at most~$\zeta$.

\textbf{Lower Bound of the Problem.}
Let $\zeta$ be the independent kissing number of a graph $G=(V,E)$. So, there exists a vertex $v\in V$ such that  $\varphi(G[N(v)])=\zeta$. Recall that $\varphi(G)$  denotes the size of a maximum independent set of a graph $G$. In other words, there exists a set $I\subseteq N(v)$ of vertices such that $I$ is an independent set of size $\zeta$. Let $I=\{v_1,v_2,\ldots,v_{\zeta}\}$. 
For the sequence of vertices   $v_1,v_2,\ldots,v_{\zeta}, v$  any online algorithm will report the first $\zeta$ vertices as a dominating set; whereas the optimum dominating set contains only the last vertex $v$.  Thus, the lower bound of the competitive ratio for the minimum dominating set is at least~$\zeta$.
 \end{proof}

As a result of Observation~\ref{obs1}, the output produced by $\AGDS$ is an independent dominating set. It is easy to observe that both arguments for the upper-bound and the lower-bound construction would apply to the MIDS problem. Therefore, we have the following.

 \begin{theorem}\label{th:mids_geometric_objects}
{The algorithm $\AGDS$ has a competitive ratio of at most $\zeta$ for the MIDS problem of a graph  having an independent kissing number at most~$\zeta$. This result is tight: the competitive ratio of any online algorithm for this problem is at least $\zeta$.}
\end{theorem}


\subsection{Minimum Connected Dominating Set}\label{2.2}

{Recall that for the case of the MCDS problem, we use Relaxed-Online-Model, and the revealed induced subgraph must always be connected.
 In this setup, Eidenbenz~\cite{Eidenbenz} proposed a greedy algorithm for unit disk graphs and showed that the algorithm achieves a competitive ratio of at most $8+\epsilon$. We analyse the same algorithm  for graphs with the fixed independent kissing number~$\zeta$.}

\textit{Description of Algorithm $\AGCDS$:} Let $V$ be the set of vertices presented to the algorithm and $\A \subseteq V$ be the set of vertices chosen by our algorithm such that $\A$ is a connected dominating set for the vertices in $V$. The algorithm maintains two disjoint sets $\A_1$ and $\A_2$ such that $\A= \A_1 \cup \A_2$. Initially, both $\A_1,\A_2=\emptyset$. Let $v$ be a new vertex presented to the algorithm. The algorithm first updates $V\leftarrow V\cup \{v\}$ and then does the following. 
\begin{itemize}
  \item If $v$ is dominated by the set $\A$, do nothing.
  \item Otherwise,  first,   add $v$ to $\A_1$.  If $v$ has at least one neighbour in $V$,  choose any one  neighbour, say $u$, of $v$ from $V$, and  add  $u$ to $\A_2$. In other words,  update $\A_1\leftarrow\A_1 \cup \{v\}$ and if necessary update $\A_2\leftarrow\A_2 \cup \{u\}$. Note that $u$ is already dominated by the existing dominating set $\A$. As a result, if we add $u$ to the dominating set, it will result in a connected dominating set.
  \end{itemize}


Note that the algorithm produces a feasible connected dominating set. The addition of vertex in $\A_1$ assures that $\A$ is a dominating set, and the addition of vertices in $\A_2$ ensures that  $\A$ is a connected dominating set. 
Now, using induction, it is easy to prove the following.

\begin{lemma}\label{lem:inva}
The algorithm $\AGCDS$ maintains  the following two invariants: (i) $\A_1$ is an independent set, and (ii)  $|\A_1|\geq|\A_2|$.
\end{lemma}

\begin{proof}
Note that whenever we add any vertex to the set $\A_2$, we are also adding a vertex to $\A_1$. Therefore, the invariant~(ii) is maintained throughout the algorithm.
Now, we prove by induction that invariant~(i) is also maintained. It is easy to observe that the invariant trivially holds at the beginning. Let us assume that the invariant holds before the arrival of the $i$th vertex $v$. 
If vertex $v$ is already dominated by the existing dominating set $\A$, then we have nothing to prove. Without loss of generality, let us assume that $v$ is not dominated by the existing dominating set $\A$. Since $v$ is not dominated by $\A$,  it is not adjacent to any of the vertices in $\A$, particularly to $\A_1$. Since the existing set $\A_1$ is an independent set, $\A_1 \cup \{v\}$ is also an independent set. So, invariant~(i) is maintained. 
 \end{proof}

\noindent
The next lemma 
is a  generalization of  a result by Wan et al.~\cite[Lemma 9]{WanAF04}.

\begin{lemma}\label{lm:cds_opt}
Let $\cal I$ be  an independent set, and  $\OO$ be a minimum connected dominating set  of a  graph with the independent kissing number $\zeta$. Then $| {\cal I}| \leq (\zeta-1){| \OO| } +1$.
\end{lemma}
\begin{proof}
 Let $T$ be a spanning tree of $\OO$.
 Consider an arbitrary preorder traversal of $T$ given by $v_1, v_2,\ldots, v_{{| \OO| }}$. Let ${\cal I}_1$ be the set of vertices in ${\cal I}$ adjacent to $v_1$. For any $2\leq i\leq {| \OO| }$, let ${\cal I}_i$ be the set of vertices in ${\cal I}$ that are adjacent to $v_i$ but not adjacent to any of $v_1,\ v_2,\ldots,\ v_{i-1}$. Then ${\cal I}_1,\ {\cal I}_2,\ldots,\ {\cal I}_{{| \OO| }}$ forms a partition of ${\cal I}$. Since the graph has the independent kissing number $\zeta$, each vertex $v\in\OO$ can be adjacent to at most $\zeta$  pairwise non-adjacent vertices. As a result, for each $1\leq i\leq |\OO|$, $| {\cal I}_i|  \leq \zeta$. On the other hand,  for any $2\leq i\leq {| \OO| }$, at least one vertex $v_t \in \{v_1,\ v_2,\ldots,\ v_{i-1}\}$ is adjacent to $v_i$ and each vertex in ${\cal I}_i$ is adjacent to none of the vertices in  $\{v_1,\ v_2,\ldots,\ v_{i-1}\}$. So, $\{v_t\}\cup {\cal I}_i$ forms an independent set adjacent to the vertex $v_i$.
Therefore, $| {\cal I}_i|  \leq (\zeta-1)$ where $2\leq i\leq {| \OO| }$. Consequently, we have 
$| {\cal I}| = \sum_{i=1}^{{| \OO| }} | {\cal I}_i|
   = {\cal I}_1 +\sum_{i=2}^{{| \OO| }} | {\cal I}_i| 
  \leq \zeta + ({| \OO| }-1)(\zeta-1)
   \leq (\zeta-1){| \OO| } +1$.
%
Hence, the lemma follows.
 \end{proof}

Now, we present the following theorem.
\begin{theorem}\label{cds}
The algorithm $\AGCDS$ has an asymptotic competitive ratio of at most~$2(\zeta-1)$
and an absolute competitive ratio of at most~$2\zeta$ for the MCDS problem for a graph having an independent kissing number at most~$\zeta$.
\end{theorem}
\begin{proof}
Let $\A$ be the connected dominating set returned by $\AGCDS$, and let $\OO$ be a minimum connected dominating set for the input graph.
Due to Lemma~\ref{lem:inva}, we have $| \A| =|\A_1|+|\A_2|\leq 2| \A_1| $. On the other hand, we know that $\A_1$ is an independent set (due to Lemma~\ref{lem:inva}).  So, using  Lemma~\ref{lm:cds_opt}, we have $| \A_1| \leq (\zeta-1)| \OO| +1$. As a result, we get $| \A| \leq 2(\zeta-1)| \OO| +2$. Therefore, $\AGCDS$ has an asymptotic competitive ratio of at most~$2(\zeta-1)$.
Note that we have $| { \A_1}| \leq (\zeta-1){| \OO| } +1\leq \zeta {| \OO| }$. Here, the second inequality follows from the fact that ${| \OO| }\geq 1$.  
Thus, for the (absolute) competitive ratio, we have 
$\frac{| \A| }{{| \OO| }}\leq 
     \frac{2| {\A_1}| }{{|{\A_1}|}/{\zeta}} = 2\zeta$.
This completes the proof.
  \end{proof}

\begin{remark}
    Note that due to the result of Du and Du~\cite[Thm 1]{DuD15}, for unit disk graphs,
 we have  $|{\cal I}|  \leq 3.399| \OO|  + 4.874$. As a result, similar to Theorem~\ref{cds}, 
one can prove that $| \A|  \leq 6.798| \OO| +9.748$. Hence, for unit disk graphs, $\AGCDS$ has an asymptotic competitive ratio of at most~$6.798$.
\end{remark}

 \section{Lower Bound of the MCDS Problem}\label{sec:lb_MCDS}

In this section, first, we propose a lower bound of the MCDS problem for a wheel graph.  Then, using that, we propose a lower bound for the geometric intersection graph of translated copies of a convex object in $\IR^2$. 

Consider a \emph{wheel graph} $W_k=(V,E)$ of order $k$, where $V=\{v_0,v_1,\ldots, v_k\}$
and $E=\big{\{}\{v_i,v_k\}\ |\ i\in[k-1]\big{\}}\cup\big{\{}\{v_i,v_{(i+1)\mod k}\} |\ i\in[k-1]\big{\}}$. In other words, in $W_k$, the vertices $v_0,v_1,\ldots, v_{k-1}$ form a  cycle $C_k$ and a single core vertex $v_k$ is adjacent to 
each vertex of $C_k$.
 Now, we define a cyclone-order of vertices in a wheel graph~$W_k$.
\begin{figure}[htbp]
     \centering
\includegraphics[width=44 mm]{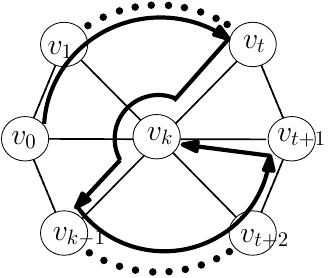}
       \caption{Cyclone-order of vertices in a wheel graph. The bold arrow indicates a cyclone-order.}
       \label{fig:new_cyclone}
\end{figure}

\begin{definition}[Cyclone-order of vertices in a wheel graph]
For an integer $t\ (0<t<k-1)$, in the the cyclone order of $W_k$,   first, we enumerate $t+1$ vertices $v_0, v_{1}, \ldots, v_{t}$ of $C_k$, followed by an enumeration of the remaining $k-t-1$ vertices of $C_k$, i.e.,   $v_{k-1}, v_{k-2},\ldots,v_{t+2},v_{t+1}$.   Finally, the core vertex $v_k$ is appended.
 We denote the first $t+1$ length sequence of $C_k$ as a \emph{cw-part} and the remaining $k-t-1$ length sequence as an \emph{acw-part} of the cyclone-order. 
\end{definition}

Now, it is easy to obtain the following lemma.
\begin{lemma}\label{lemma:path-of-cycles}
 If the vertices of a wheel graph $W_k$ are enumerated in a cyclone-order,  then any online algorithm reports a CDS of size at least~$k-2$, where the size of an offline optimum is 1.  
\end{lemma}
\begin{proof}
Consider an input sequence of vertices of $W_k$, where the vertices arrive in cyclone-order. First, let us consider the cw-part of the sequence  $v_0, v_{1}, \ldots, v_{t}$. 
 Any online algorithm will report at least the first $t$ vertices as a CDS for this sequence. Next, if we append the first $k-t-2$ vertices of acw-part $v_{k-1}, v_{k-2},\ldots, v_{t+2}$, then any online algorithm will report at least  $k-3$  objects as CDS in total. Now, if we append $v_{t+1}$ to the sequence, any online algorithm will add $v_{t}$ or $v_{t+2}$ in the CDS.  Therefore, any online algorithm reports a CDS of size at least  $k-2$. 
Since $v_k$ is the core vertex, it is already dominated. Hence, if we append $v_{k}$ to the sequence, any online algorithm reports a CDS of size at least  $k-2$; while an offline optimum will report $v_{k}$ as a CDS.  Therefore, for $W_k$, any online algorithm reports a CDS of size at least  $k-2$; while the offline optimum is 1.
 \end{proof}
Now, we give an explicit construction of a wheel graph $W_{2\zeta}$ using translates of a convex object having independent kissing number $\zeta$.

\begin{lemma} \label{lm:C_block}
For a family of translates of a convex object having independent kissing number $\zeta$, there exists a geometric intersection graph $W_{2\zeta}$.
\end{lemma}

{Before presenting the proof of the of Lemma~\ref{lm:C_block}, first, we present an observation, followed by a lemma, in which we prove that for a family of translated convex objects, we can always have an optimal independent kissing configuration in a standard form. }

\begin{observation}\label{obs:0}
Let $O$ and $L$ be a convex object and a line, respectively. Let $x$ be any point on $L$. If $x$ moves along the  line $L$ from one end to another, then   the distance $d_C(x, O)$ first decreases monotonically until a minimum point (or an interval of the same minimum value), then increases monotonically (refer to Figure~\ref{fig:convex_dist}).
\end{observation}

Now, consider the following lemma.
\begin{lemma}\label{lm:standardIndepndentKC}
 There exists a standard optimal independent kissing configuration for a family of translated copies of a convex object in $\IR^2$.
\end{lemma}


\begin{figure}[!ht]
     \begin{subfigure}[b]{0.45\textwidth}
            \includegraphics[scale=.33]{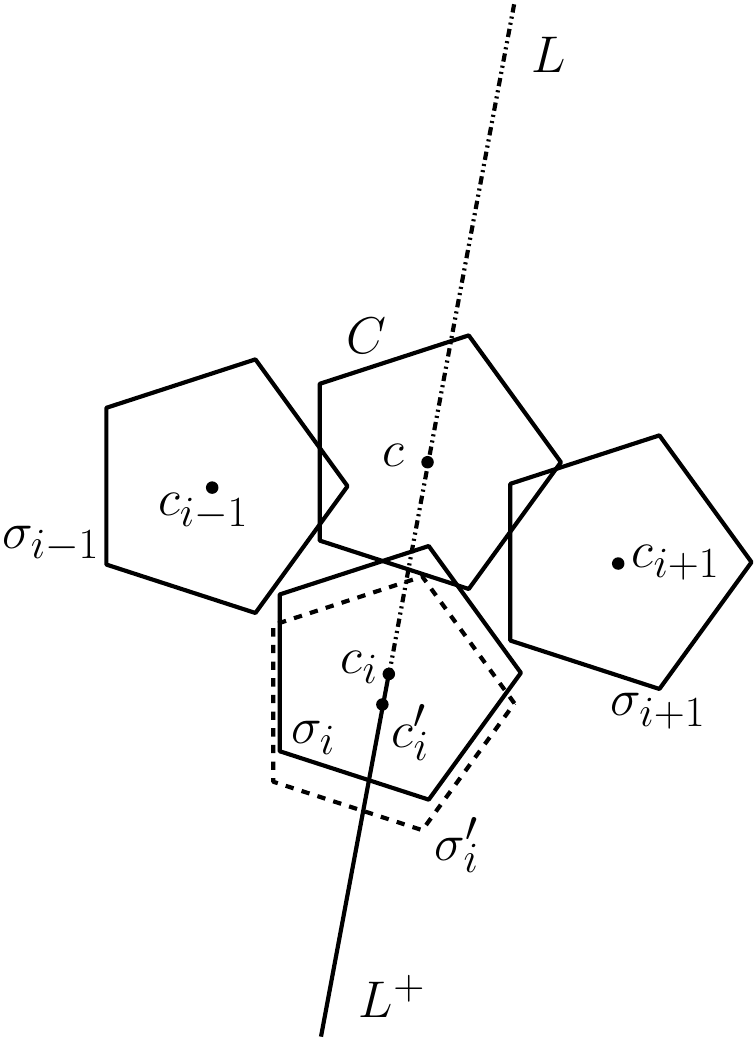} 
            \caption{}
            \label{fig:ikn_to_sikn}
     \end{subfigure}
     \hfill
     \begin{subfigure}[b]{0.45\textwidth}
        \includegraphics[scale=.43]{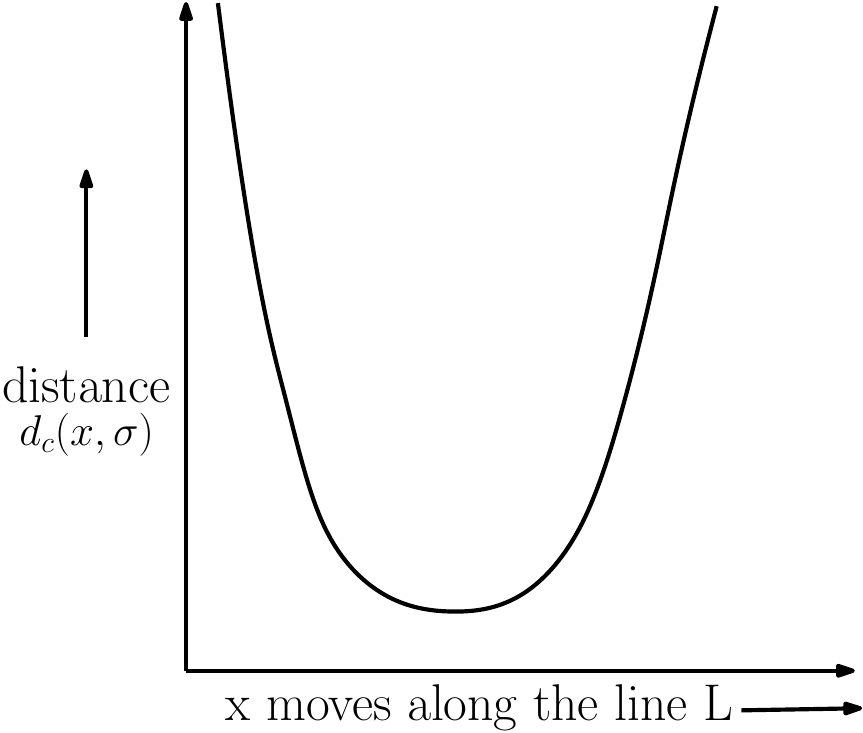}
        \caption{}
        \label{fig:convex_dist}
     \end{subfigure}
      \caption{(a)~Illustration of Lemma~\ref{lm:standardIndepndentKC}: Moving $\sigma_i$ along $L^{+}$; (b)~Illustration of Observation~\ref{obs:0}.}
      \label{fig:translates}
\end{figure}

\begin{proof}
To prove the lemma, we show that for any given optimal independent kissing configuration, we can transform it into a standard form. Consider an optimal independent kissing configuration, where $C$ is the core, and $\cal I$ is the independent set. Let  $({\sigma}_0$, ${\sigma}_1$,\ldots, ${\sigma}_{\zeta-1})$ be the order of appearance of objects in $\cal I$ around the core object $C$ in the configuration, where ${\sigma}_i\in {\cal I}$ and $i\in \mathbb{Z}_{\zeta}$. Consider an object ${\sigma}_i\in {\cal I}$  with a nonempty common interior with $C$. Let $c$ and $c_i$  be the center of the object $C$ and ${\sigma}_i$, respectively (see Figure~\ref{fig:ikn_to_sikn}). Let $L$ be a line obtained by extending the line segment $cc_i$ in both directions.  Since objects ${\sigma}_i$ and ${\sigma}_{i-1}$ (respectively, ${\sigma}_{i+1}$) are non-touching and ${\sigma}_i$ intersects $C$, we have  $d_C(c_i,{\sigma}_{i-1}) > d_C(c,{\sigma}_{i-1})$ (respectively, $d_C(c_i,{\sigma}_{i+1}) > d_C(c,{\sigma}_{i+1})$). Note that the operation on index $i$ is modulo $\zeta$. Consider the half-line $L^{+}$  of $L$  with one endpoint at $c_i$ and does not contain the point $c$. Due to Observation~\ref{obs:0}, if $x$ moves along $L^{+}$  from $c_i$ then the distance $d_C(x, {\sigma}_{i-1})$ (respectively, $d_C(x, {\sigma}_{i+1})$) monotonically increases. As a result, if we translate  the object ${\sigma}_i$ by moving the center along $L^{+}$, the translated copy ${\sigma}_i'$ never touches the object ${\sigma}_{i-1}$ (respectively, ${\sigma}_{i+1}$). Hence, we will be able to find a point $x$  on $L^{+}$ such that the translated copy ${\sigma}_i'$ centering on $x$ will touch the boundary of $C$,  without touching  ${\sigma}_{i-1}$ and ${\sigma}_{i+1}$.  In the independent kissing configuration, we replace the object ${\sigma}_i$ with ${\sigma}_i'$. For each object $I\in {\cal I}$ that has a nonempty common interior with $C$,
we follow a similar approach as above. This ensures that all these $\zeta$ objects are mutually non-overlapping (but touching) with $C$. Hence, the lemma follows.\end{proof}

Now, we present the proof of Lemma~\ref{lm:C_block}.

\noindent
\textbf{Proof of Lemma~\ref{lm:C_block}}

\noindent
The proof is by construction. 
  Due to Lemma~\ref{lm:standardIndepndentKC}, we know a standard optimal independent kissing configuration exists for a family of translated copies of a convex object. Let $K$ be a standard independent kissing configuration, where $\sigma$ is a core object, and $\cal I$ is an independent set. Let  $({\sigma}_0$, ${\sigma}_1$,\ldots, ${\sigma}_{\zeta-1})$ be  the order of appearance of objects in $\cal I$  around the core object $\sigma$ in the configuration, where ${\sigma}_i\in {\cal I}$, $i\in \mathbb{Z}_{\zeta}$. Let us define a locus $\cal L$ that contains all points $x\in \IR^2$ such that $d_{\sigma}(x,\sigma)=1$. Note that the centers of all objects $\sigma_i\in {\cal I}$ lie on the locus $\cal L$. Now, for each $i\in \mathbb{Z}_{\zeta}$,  we  make a copy $\sigma_i'$ of $\sigma_i$. {We translate $\sigma_i'$ around $\sigma$ keeping the center of $\sigma_i'$  on the locus $\cal L$ until $\sigma_i'$ touches the object $\sigma_{i+1}$.}
 Note that  $\sigma_i'$  also touches $\sigma_i$; otherwise, $K$ is not an optimal configuration since we can place an extra object between $\sigma_i$ and $\sigma_{i+1}$. It is easy to observe that, apart from these two objects ($\sigma_i$ and $\sigma_{i+1}$), $\sigma_i'$ does not intersect any other objects in  ${\cal I}$.
 In this way, we obtain a set ${\cal I}'=\{{\sigma}_0'$, ${\sigma}_1'$,\ldots, ${\sigma}_{\zeta-1}'\}$ of $\zeta$ objects. The method of construction ensures that ${\cal I}'$ is an independent set (and ${\cal I}'$ together with $\sigma$ is a standard optimal independent kissing configuration). Now, consider the order $({\sigma}_0, {\sigma}_0', {\sigma}_1, {\sigma}_1',\ldots, {\sigma}_{\zeta-1}, {\sigma}_{\zeta-1}')$ of
 appearance of  objects in $\cal I\cup {\cal I}'$  around $\sigma$. Here, each object is intersected by exactly two objects: its previous and next object in the sequence. Therefore, the geometric intersection graph of $\cal I\cup {\cal I}'\cup\sigma$ is a $W_{2\zeta}$, where $\sigma$ is the core object.
 \myqed\\

Combining Lemmas~\ref{lemma:path-of-cycles} and~\ref{lm:C_block}, we have the following result.
\begin{theorem}\label{thm:mcds_con}
Let $\zeta$ be the independent kissing number of a family $\cal S$ of translated copies of a convex object in $\IR^2$.
Then the competitive ratio of every online algorithm for MCDS of $\cal S$ is at least $2(\zeta-1)$.
\end{theorem}




\section{Algorithm for the Minimum Coloring Problem}\label{sec:coloring}

In this section, we consider the minimum coloring problem for the geometric intersection graph of bounded scaled $\alpha$-fat objects in $\IR^d$ having a width between $[1,m]$, where $m\geq 2$. In the subsequent subsections, first, we present a deterministic algorithm $\AL$; next, we present an analysis of the well-known algorithm  $\AFF$ for this problem.  While the algorithm $\AL$ uses the geometric representation, the algorithm $\AFF$ does not use it.
Note that the algorithm $\AL$ uses $\AFF$ as a subroutine.
Here,  we mention  $\AFF$  and its known competitive ratio for graphs with bounded $\zeta$.

\textit{Description of the algorithm $\AFF$}:
Upon the arrival of an object $\sigma$, the algorithm $\AFF$ assigns the smallest color available, i.e., the smallest positive integer that has not yet been assigned to an adjacent vertex of $\sigma$.

\begin{lemma}~\cite[Lemma~4]{CapponiP05}\label{lem:coloring}
   The algorithm  $\AFF$ has a competitive ratio of $\zeta$ for the MC problem for a graph having an independent kissing number at most~$\zeta$.
\end{lemma}

\subsection{Algorithm $\AL$}
This algorithm is similar to the algorithm of Erlebach and Fiala~\cite{ErlebachF02} originally defined for bounded scaled disks having radii in the interval $[1,m]$.
For any $j\in\mathbb{Z}^{+}\cup\{0\}$, let $L_j$ be the $j$th layer containing all objects with widths in the interval $[2^j,2^{j+1})$. 
Observe that the width of each layer's objects falls within a factor of two.
For each layer $L_j$, we use $\AFF$ separately to color the objects. When an object $\sigma_i$ having width $w_i$ arrives, our algorithm, first, determines the layer number $j=\lfloor\log w_i\rfloor$.  Then we color $\sigma_i$ using
$\AFF$ considering already arrived objects in  $L_j$,
 and also we use the fact that
a color that is used in any other layer cannot be used for $\sigma_i$. 
A pseudo-code of the algorithm $\AL$ is given in Algorithm~\ref{alg:cap}.
\begin{algorithm}[htbp]
\caption{$\AL$}\label{alg:cap}
\footnotesize{\begin{algorithmic}[1]
 \For {$i= 1$ to $\infty$;} \Comment{Arrival of an object $\sigma_i$ having a width $w_i$}
 \Begin
\State $j=\lfloor \log_2 w_i\rfloor$; \Comment{Identifying the index of the layer to which $\sigma_i$ belongs, where $r_i$ is the width of $\sigma_i$.}

\If {$\sigma_i$ is the first arrived object from layer $L_j$}
 \State $L_j\leftarrow \emptyset$
\EndIf
\State $L_j\leftarrow L_j\cup\{\sigma_i\}$; \Comment{The layer containing $\sigma_i$}
\State $F=\{c(\sigma_k):1\leq k<i, \sigma_k\in L_j, \sigma_k\cap\sigma_i\neq\emptyset\}\cup\{c(\sigma_k):1\leq k<i, \sigma_k\notin L_j\}$; \Comment{The set of forbidden colors}
\State $c(\sigma_i)=\min\{\mathbb{Z}^{+}\setminus F\}$; \Comment{color assigned to $\sigma_i$}
 \EndBegin
 \EndFor
\end{algorithmic}}
\end{algorithm} 
\begin{theorem}\label{thm:AL}
   Let $\zeta'$ be the independent kissing number of bounded scaled $\alpha$-fat objects having widths in the interval $[1,2]$. The algorithm $\AL$ has a competitive ratio of at most $\zeta'(\lfloor\log m \rfloor+1)$ for MC of geometric intersection graph of bounded scaled $\alpha$-fat objects in $\IR^d$ having widths in the interval $[1,m]$.
\end{theorem}
\begin{proof}
 Let $\A$ and $\OO$ be the set of colors used by the algorithm $\AL$ and offline optimum for an input sequence $\cal I$. For each $i\in\{0,1,\ldots,\lfloor \log m\rfloor\}$, let the layer $L_i$ be the collection of all $\alpha$-fat objects in $\cal I$ having widths in $\left[2^i, 2^{i+1}\right)$. Let $\OO_i$ be a set of colors the offline optimum algorithm uses for the layer $L_i$. Let $\OO_i'\subseteq\OO$ be the set of colors used for the layer $L_i$.
 Note that the colors in $\OO_i'$ are a valid coloring for objects in $L_i$. Thus, we have $|\OO_i|\leq |\OO_i'|$. 
 Let $\A_i$ be the set of colors used by the algorithm $\AL$ to color layer $L_i$. Note that $\A=\cup_{i=0}^{\lfloor \log m\rfloor}\A_i$ and $\A_i\cap\A_j=\emptyset$, where $i\neq j\in[\lfloor\log m\rfloor]$.
  Due to Lemma~\ref{lem:coloring}, for all $i\in[\lfloor\log m\rfloor]$, we have $|\A_i|\leq{\zeta_{i}}|\OO_i|$, where $\zeta_{i}$ is the independent kissing number of bounded scaled $\alpha$-fat objects having widths in $[2^i,2^{i+1})$. Since the width of objects in each layer is within a factor of two, for each $i$, the value of $\zeta_{i}$ is the same as $\zeta'$. Since $|\OO_i|\leq |\OO_i'|\leq|\OO|$, we have $|\A_i|\leq {\zeta'}|\OO_i'|\leq 
  {\zeta'}|\OO|$.   Then, we have $|\A|=\sum_{i=0}^{\lfloor \log m\rfloor}|\A_i|\leq\sum_{i=0}^{\lfloor \log m\rfloor}{\zeta'}|\OO_i|=\zeta'(\lfloor\log m \rfloor+1)|\OO|$.
Hence, the theorem follows.
 \end{proof}

\subsection{Algorithm $\AFF$}
 Recall that the algorithm $\AFF$ does not require object representation in advance. Here, we show that, for MC of geometric intersection graph of bounded scaled $\alpha$-fat objects having widths in the interval $[1,m]$, the competitive ratio achieved by the algorithm $\AFF$ is asymptotically the same as that attained by the algorithm $\AL$. The analysis is a generalization and very similar to~\cite{CaragiannisFKP07b}, except  Claim~\ref{clm:layer}. For the sake of completeness, we have described it.

\begin{theorem}\label{thm:AFF}
    Let $\zeta'$ be the independent kissing number of bounded scaled $\alpha$-fat objects having widths in the interval $[1,2]$. The algorithm \textsc{First-Fit} has an asymptotic competitive ratio of $O(\zeta'\log m)$ for MC of geometric intersection graph of bounded scaled $\alpha$-fat objects having widths in the interval~$[1,m]$.
\end{theorem}

\begin{proof}
 Let $\kappa$ be the number of colors used by an offline optimum for an input sequence $\cal I$.  An object $o\in\I$, having width $w_o$, belongs to layer $L_i$, where $i=\lfloor\log w_o\rfloor$.
For the sake of simplicity, throughout the proof, we will use the layer index $i$ to denote the layer $L_i$. 
 Note that all the objects in $\I$ must belong to some layer $i\in\{0,1,\ldots,\lfloor \log m\rfloor\}$. 
 We first prove the following claim.
 \begin{claim}\label{clm:layer}
     An object belonging to layer $i\geq 0$ can intersect at most $\zeta'(\kappa-1)$ other objects of layers at least $i$. 
 \end{claim}
 \begin{proof}
     For a contradiction, let us assume that there exists an object $o$ belonging to layer $i$ such that the object $o$ intersects at least $\zeta'(\kappa-1)+1$ other objects of layers at least $i$. 
     Let ${\cal S} \subseteq {\cal I}$ be the collection of all objects belonging to layers at least $i$ that intersect $o$. Since all the objects in ${\cal S}$ intersect with $o$, and 
 the offline optimum requires at most $\kappa$ colors to color $\cal I$, the offline optimum needs at most $\kappa-1$ colors to color the objects in ${\cal S}$.  Since, as per our assumption, $|{\cal S}| \geq \zeta'(\kappa-1)+1$, there exists a subset ${\cal S}'\subseteq  {\cal S}$ of cardinality at least $\zeta'+1$ such that all of the objects in ${\cal S}'$ are assigned a same color by the offline optimum. Note that ${\cal S}'$ consists of pairwise non-touching objects.

       It is evident that every object in ${\cal S}'$ possesses a minimum width of $2^i$, and the width of the object $o$ is in the interval $[2^i, 2^{i+1})$. Now, we shrink (scale down) each object $\sigma\in {\cal S}'$ into an object $\sigma'$ of width $2^i$ such that the object $\sigma'$  is totally contained in $\sigma$, and the object $\sigma'$ still intersects the object $o$. 
 Let ${\cal S}''$ be the collection of all such shrunk objects.
  Since all the objects in ${\cal S}''$ intersect with $o$ and the objects in ${\cal S}''$ are pairwise non-touching, the set ${\cal S}''$  forms an independent kissing configuration. Thus, the independent kissing number for the set ${\cal S}'' \cup \{o\}$ is at least $ |{\cal S}''|=|{\cal S}'|\geq\zeta'+1$. On the other hand, since the width of objects in ${\cal S}'' \cup \{o\}$ is within a factor of two, the independent kissing number for the set ${\cal S}'' \cup \{o\}$ is at most
  $\zeta'$. Thus, we get a contradiction.
  \end{proof}
 
  Now, we will show that each object $o\in {\cal I}$ at  layer $i\geq 0$ is colored by the algorithm $\AFF$ with a color from the range $[1, (\zeta'(\kappa-1)+1)(i+1)]$. Thus, the maximum color that the algorithm $\AFF$ can use is~$(\zeta'(\kappa-1)+1)(\lfloor\log m\rfloor+1)$. Since $\kappa$ is the offline optimum, the asymptotic competitive ratio obtained by the algorithm $\AFF$ is $O(\zeta'\log m)$.

By using induction, we show that each object $o\in {\cal I}$ at layer $i\geq 0$ is colored by the algorithm $\AFF$ with a color from the range $[1, (\zeta'(\kappa-1)+1)(i+1)]$. Due to claim~\ref{clm:layer}, the induction statement is true for objects at layer~0. Assume that the statement holds true for objects at the layer $i=0,1,2,\ldots,k\ (<\lfloor\log m\rfloor)$. We will now demonstrate that the statement holds for $i=k+1$. let us consider an object $o$ at layer $k+1$. Note that, due to the induction hypothesis, each object intersecting $o$  from the lower layers less than $k+1$ are colored using colors from the range $[1, (\zeta'(\kappa-1)+1)(k+1)]$. On the other hand, due to Claim~\ref{clm:layer}, the object $o$ can intersect at most $\zeta'(\kappa-1)$ additional objects from layers at least $k+1$. Thus, the maximum colors that can be assigned to object $o$ by the algorithm $\AFF$  is $(\zeta'(\kappa-1)+1)(k+1)$+$\zeta'(\kappa-1)$+1=$(\zeta'(\kappa-1)+1)(k+2)$. This completes the proof.
\end{proof}




\section{Value of Independent Kissing Number for Families of Geometric Objects}\label{sec:IKS}
Note that the value of $\zeta$ for unit disk graphs is already known to be 5~\cite{Eidenbenz}. Here, we study the value of $\zeta$ for other geometric intersection graphs.
\begin{theorem}\label{thm:main_zeta}
 The independent kissing number for the family of
 \begin{itemize}
 \item [(a)] congruent balls in $\IR^3$ is~12;
     \item[(b)] translated copies of a hypercube in $\IR^d$ is $2^d$, where $d\in\mathbb{Z}^{+}$;
     \item[(c)] translated copies of an equilateral triangle is at least $5$ and at most $6$;
     \item[(d)] translated copies of a regular $k$-gon $(k\geq 5)$ is at least $5$ and at most $6$;
   \item [(e)] congruent hypercubes in $\IR^d$ is at least~$2^{d+1}$, where $d\geq 2$ is an integer;
    \item[(f)] bounded scaled disks in $\IR^2$ having radii in the interval $[1,2]$ is~11;
     \item[(g)] bounded-scaled $\alpha$-fat objects in $\IR^d$ having widths in the interval $[1,m]$ is at least~$\left(\frac{\alpha}{2}\left(\frac{m+2}{1+\epsilon}\right)\right)^d$ and at most~$\Big(\frac{m}{\alpha}+2\Big)^d$, where $\epsilon>0$ is a very small constant and $d\in\mathbb{Z}^{+}$.
    
     \end{itemize}
 \end{theorem}
For each item (except (e)) of the above theorem, we prove both the upper and lower bounds of the value of the independent kissing number. \\


 \noindent
\textbf{Theorem~\ref{thm:main_zeta}(a).} \textit{The independent kissing number for the family of congruent balls in $\IR^3$ is~12.}
 \begin{proof}
First, we present a lower bound.
Let $C$ be  a regular icosahedron whose each edge is of length $\ell$=$2+\epsilon$, where $0<\epsilon <1$ (in particular, one can choose $\epsilon$=$0.001$).
Let corner points of $C$ be the centers  of  unit (radius) balls $\sigma_1,\sigma_2,\ldots,\sigma_{12}$. 
Since the edge length of the icosahedron $C$ is greater than 2, all these balls are mutually non-touching. Let $B$ be the circumscribed ball of the icosahedron $C$. It is a well-known fact that if a regular icosahedron has edge length $\ell$, then the radius $r$ of the circumscribed ball is $r=\ell \sin(\frac{2\pi}{5})$~\cite{novivant}. In our case, it is easy to see that  $r<2$.  In other words, the distance from the center of $B$ to each corner point of $C$ is less than two units.
Thus, each of these unit balls $\sigma_1,\sigma_2,\ldots,\sigma_{12}$ is intersected by a unit ball $\sigma_{13}$ whose center  coincides with the center of  $B$. This implies that the value of $\zeta$ for congruent balls in $\IR^3$ is at least~12.
The upper bound follows from the fact that the kissing number for balls in $\IR^3$ is 12~\cite{BrassMP,Schutte}. Hence,  the result follows.
 \end{proof}


\noindent
\textbf{Theorem~\ref{thm:main_zeta}(b).}  \textit{The independent kissing number for the family of translated copies of a hypercube in $\IR^d$ is $2^d$, where $d\in\mathbb{Z}^{+}$.}


\begin{proof}
\textbf{Upper Bound.} Let $K$ be an optimal independent kissing configuration for translates of an axis-parallel unit hypercube in $\IR^d$. Let the core of the configuration be $u$. 
It is easy to observe that an axis-parallel hypercube $R$, with a side length of 2 units, contains all the centers of hypercubes in $K\setminus \{u\}$.  Let us partition $R$ into $2^{d}$ smaller axis-parallel hypercubes, each having unit side length.
Note that each of these smaller hypercubes can contain at most one center of a hypercube in $K\setminus \{u\}$. As a result, we have $| K\setminus \{u\}|  \leq 2^d$. Therefore, the independent kissing number for translates of a hypercube in $\IR^d$ is at most $2^{d}$.

\textbf{Lower Bound.} We give an explicit construction of an independent kissing configuration $K$, where the size of the independent set is $2^d$.
 Let $\sigma_1,\sigma_2,\ldots,\sigma_{2^d}$  and $\sigma_{2^d+1}$ be the $d$-dimensional  axis-parallel unit hypercubes of $K$.
 We use $c_i$ to denote the center of $\sigma_i$, for $i\in [2^d+1]$. Let the center  $c_{2^d+1}=(\frac{1}{2},\frac{1}{2},\ldots,\frac{1}{2}$),  and $p_1,p_2,\ldots,p_{2^d}\in \IR^d$ be corner points of the  hypercube $\sigma_{2^d+1}$. It is easy to observe that each coordinate of $p_i, i\in [2^d]$ is either 0 or 1.
Let $\epsilon$ be a positive constant satisfying $0<\epsilon <\frac{1}{2\sqrt{d}}$. For $i\in[2^d]$ and $j\in[d]$, let us define the $j$th coordinates of $c_i$ as follows:
  \begin{equation}\label{eq_no_1}\small{
c_i(x_j) =
    \begin{cases}
      -\epsilon, & \text{ if $p_i(x_j)=0$}\\
      1+\epsilon, & \text{ if $p_i(x_j)=1$},  
    \end{cases}   }   
\end{equation}
where $c_i(x_j)$ and $p_i(x_j)$ are the $j^{th}$ coordinate value of $c_i$ and $p_i$, respectively.

To complete the proof, we argue that
 the hypercubes $\sigma_1,\sigma_2,\ldots,\sigma_{2^d}$ are mutually non-touching, and  each intersected by the  hypercube   $\sigma_{2^d+1}$.
To see this, first, 
note that, for any $i\in [2^d]$,  the Euclidean distance $d(p_i,c_i)$ between $p_i$  and $c_i$ is $\sqrt{d}\epsilon$ (follows from Equation~\ref{eq_no_1}). Since $\epsilon <\frac{1}{2\sqrt{d}}$, we have $d(p_i,c_i)< \frac{1}{2}$.  As a result,  the corner point $p_i$ of $\sigma_{2^d+1}$ is contained in the  hypercube $\sigma_i$. Thus,   $\sigma_{2^d+1}$ intersects $\sigma_i$, $\forall$ $i\in [2^d]$.
Now, consider any  $i,j\in [2^d]$ such that $i \neq j$. Since $p_i$ and  $p_j$
are distinct, it is easy to note that they will differ in at least one coordinate. As a result, 
the distance (under $L_{\infty}$-norm) between $c_{i}$ and $c_j$ is $(1+2\epsilon)$. So, $\sigma_i$ and $\sigma_j$ are non-touching.
 Hence, the result follows.
\end{proof}


\noindent
\textbf{Theorem~\ref{thm:main_zeta}(c).} \textit{The independent kissing number for the family of translated copies of an equilateral triangle is at least $5$ and at most $6$.}

\begin{proof}
Throughout the following proof, if not explicitly mentioned, we refer to an equilateral triangle as a unit triangle. 

For the lower bound, refer to Figure~\ref{fig:triangle}. Here, we have constructed an example of the independent kissing configuration for unit triangles, where the size of the independent set is five. The rest of the proof is for the upper bound.
\begin{figure}[htbp]
  \centering
     \begin{subfigure}[b]{0.30\textwidth}
          \centering
        \includegraphics[scale=0.37]{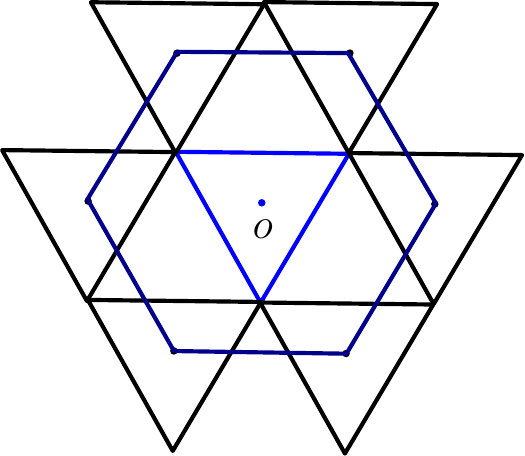}
    \caption{}
    \label{fig:nbd_tri}
     \end{subfigure}
     \hfill
     \begin{subfigure}[b]{0.30\textwidth}
          \centering
        \includegraphics[scale=0.6]{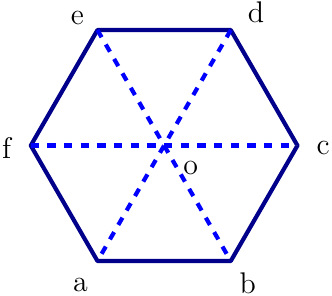}
    \caption{}
    \label{fig:tri_nbd}
     \end{subfigure}
      \hfill
     \begin{subfigure}[b]{0.30\textwidth}
          \centering
        \includegraphics[scale=0.6]{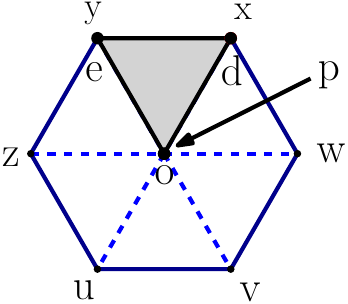}
    \caption{}
    \label{fig:tri_ode}
     \end{subfigure}
       \caption{(a) The region $N(\sigma(o))$ containing all the centers of unit triangles that can intersect  $\sigma(o)$ is marked blue; (b) The region $N(\sigma(o))$ partitioned into six symmetrical triangles; (c) The point $p$ coincides with the point $o$. } \end{figure}
 
Let $K$ be an optimal independent kissing configuration for  translates of an equilateral triangle. Let the core of the configuration be the triangle $\sigma(o)$.
It is easy to observe that if we move a translated copy $\sigma'$ of $\sigma(o)$ along the boundary of $\sigma(o)$, the locus of the center of $\sigma'$ will form a hexagon of side length one. Therefore, the hexagonal region $N(\sigma(o))$, as shown in Figure~\ref{fig:nbd_tri}, contains  all the centers of translates of a unit triangle that can touch/intersect $\sigma(o)$.

Let us partition the region $N(\sigma(o))$, denoted by the hexagon $abcdef$, into six sub-regions, each consisting of an equilateral triangle with unit side as shown in the Figure~\ref{fig:tri_nbd}.
We will show that each sub-region can contain at most one triangle that belongs to $K\setminus \{\sigma(o)\}$.
We prove this for the triangular sub-region $\Delta ode$ (proof for all other sub-regions will be similar).
To prove this, it is enough to show that if we put an equilateral triangle $\sigma(p) \in K\setminus \{\sigma\}$ for any point $p\in \Delta ode$, $N(\sigma(p))$ will cover the entire triangular sub-region $\Delta ode$.
We denote the hexagon $N(\sigma(p))$ by $uvwxyz$.

 \begin{figure}[htbp]
  \centering
     \begin{subfigure}[b]{0.3\textwidth}
          \centering
        \includegraphics[scale=0.55]{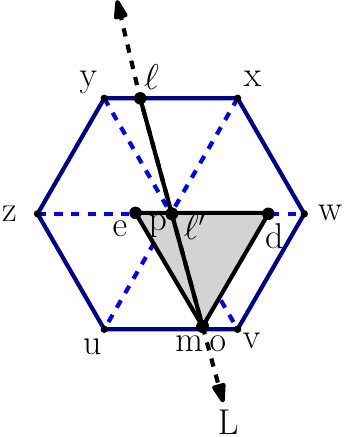}
    \caption{}
     \label{fig:tri_case1}
     \end{subfigure}
     \hfill
     \begin{subfigure}[b]{0.3\textwidth}
          \centering
        \includegraphics[scale=0.55]{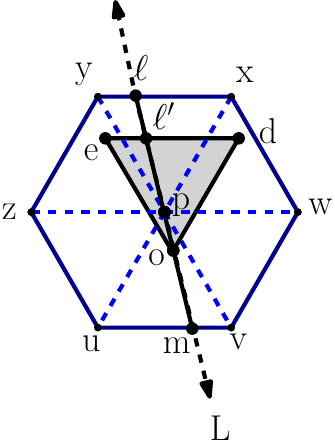}
    \caption{}
    \label{fig:tri_case2}
     \end{subfigure}
     \centering 
     \begin{subfigure}[b]{0.35\textwidth}
         \centering
\includegraphics[width=40mm]{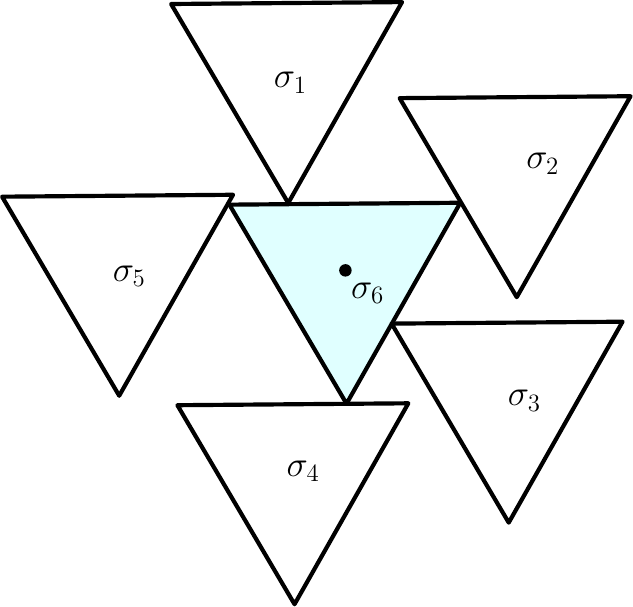}
\caption{}
\label{fig:triangle}
     \end{subfigure}
     \caption{(a) $\ell,{\ell}'$ and $m$ are the points of intersection of the line $L$ with line segments $\overline{xy}, \overline{de}$ and $\overline{uv}$, respectively; (b) The point $p$ lie inside the $\Delta ode$; (c) An independent kissing configuration for unit triangles.}
\end{figure}

\begin{itemize}
\item   [Case I] The point $p$ coincides with the point {$o$}. 
 It is easy to observe that $N(\sigma(p))$ is entirely covering triangle \( \Delta ode \) (see Figure~\ref{fig:tri_ode}).

\item [Case II]  The point $p$ lies on the line segment $\overline{od}$ or $\overline{oe}$.
Without loss of generality, let us assume that the point $p$ lies on the line segment $\overline{od}$. One way to view this is as follows: initially, as in case I, the point $p$ coincides with the point $o$; next, we move the point $p$ along the line segment $\overline{od}$ towards the point $d$ (until it reaches its new destination). We can alternatively visualize this as follows: instead of moving the point $p$ on $\overline{od}$ towards $d$, fix the point $p$ and the region $N(\sigma(p))$; move the triangle \( \Delta ode \) from its initial position  such that the point $o$ always lies on the line segment $\overline{pu}$ and $o$ moves towards $u$.
Since  $\overline{pu}$ and $\overline{yz}$ are parallel and are of the same length, $e$ will move along the line segment $\overline{yz}$ towards the point $z$ and when $o$ reaches $u$ the triangle \( \Delta ode \) is moved to the position of triangle \( \Delta upz \).
Hence, $N(\sigma(p))$ completely covers the triangle \( \Delta ode \).

\item [Case III] The point $p$ lies inside the triangle \( \Delta ode \). 
Let $L$ be a line obtained by extending the line segment $\overline{op}$ in both directions. Let $\ell,{\ell}'$ and $m$ be the points of intersection of the line $L$ with line segments $\overline{xy}, \overline{de}$ and $\overline{uv}$, respectively ({see Figure}~\ref{fig:tri_case1}).  One way to visualize the situation is as follows: initially, as in case I, the point $p$ coincides with the point $o$; next, we move the point $p$ along the line $L$ towards the point ${\ell}'$ (until it reaches its new destination). We can alternatively visualize this as follows: instead of moving $p$, fix the point $p$ and the region $N(\sigma(p))$; move the triangle $\Delta ode$ from its initial position $\Delta pxy$ such that the point $o$ moves along the line $L$ towards the point $m$. Observe that, when
$o$ reaches the boundary $\overline{uv}$, the point
$p$ coincides with the point ${\ell}'$ (see Figure~\ref{fig:tri_case2}). Thus, the triangle \( \Delta ode \) is always contained in the region $N(\sigma(p))$. This completes the proof.
\end{itemize} %
\end{proof}


\noindent
\textbf{Theorem~\ref{thm:main_zeta}(d).}  \textit{The independent kissing number for the family of translated copies of a regular $k$-gon $(k\geq 5)$ is at least $5$ and at most $6$.}


\begin{figure}[!ht]
  \centering
     \begin{subfigure}[b]{0.23\textwidth}
         \centering
\includegraphics[scale=.22]{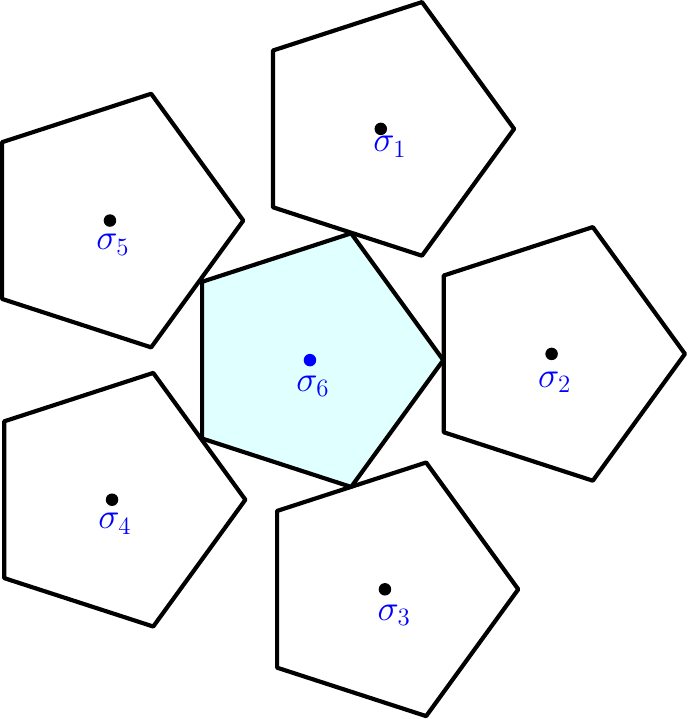} 
\caption{}
\label{fig:pentagon}
     \end{subfigure}
     \hfill
     \begin{subfigure}[b]{0.23\textwidth}
         \centering
\includegraphics[scale=.22]{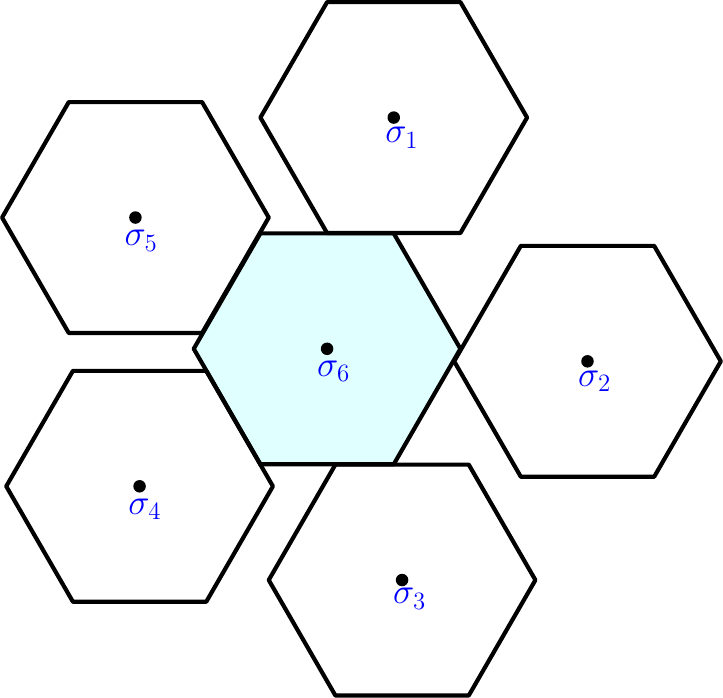}
\caption{}
\label{fig:hexagon}
     \end{subfigure}
     \begin{subfigure}[b]{0.23\textwidth}
         \centering
\includegraphics[scale=.55]{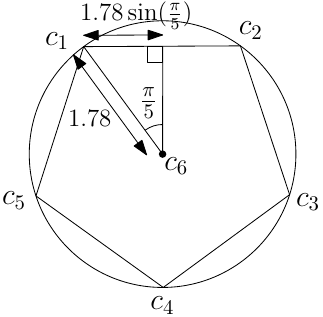}
\caption{}
\label{fig:regular_1}
     \end{subfigure}
     \begin{subfigure}[b]{0.27\textwidth}
         \centering
\includegraphics[scale=.6]{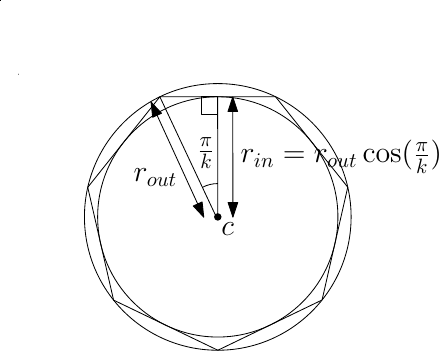}
\caption{}
\label{fig:regular}
     \end{subfigure}
       \caption{Illustration of Theorem~\ref{thm:main_zeta}(d):  (a) An independent kissing configuration for the pentagon; (b) An independent kissing configuration for  hexagon; (c) Side length of a regular pentagon inscribed in $C$; (d) Radius of the  inscribed circle of a regular $k$-gon ($k=7$).}
       \label{fig:penta_hexa}
\end{figure}
\begin{proof}

First, note that the kissing number for the regular $k$-gon 
 ($k \geq 5$) is 6~\cite{ZHAO1998293,article}.  Thus, the upper bound follows. In the rest of the proof, we prove the lower bound.
For $k=5$ and $6$, refer to  Figure~\ref{fig:penta_hexa}(a-b), where we have constructed examples of independent kissing configurations having the independent set of size five. For $k\geq 7$, we have the following proof.

Consider a circle $C(c_6)$ of radius $1.78$ (see Figure~\ref{fig:penta_hexa}(c)). Let us inscribe a regular pentagon 
$P$ of the largest side length inside $C$. 
Let $c_1, \ldots, c_5$ are the five corner points of $P$.
For each $i\in [5]$, let $\sigma_i(c_i)$ be a regular unit $k$-gon .
Observe that $dist(c_i,c_j)= 2\times 1.78 \sin(\pi/5)  > 2$, where $i,j\in[5]$ and $i\neq j$. Therefore, all the regular unit $k$-gons  $\sigma_1, \sigma_2,\ldots, \sigma_5$ are mutually non-touching. 
 Next, we argue that a regular $k$-gon $\sigma_6(c_6)$ intersects each of the
$k$-gons  $\sigma_1, \sigma_2,\ldots, \sigma_5$. 
Let $r_{out}$ and $r_{in}$ be the radius of the circumscribed circle and the inscribed circle of a regular $k$-gon, respectively. It is easy to observe that  $r_{in}=r_{out}\cos({\frac{\pi}{k}})$ (refer to Figure~\ref{fig:penta_hexa}(d)). Since $k \geq 7$, in our case the value of $r_{in}$ for each of the $\sigma_i$ is at least $\cos{\frac{\pi}{7}}(>0.9)$.
 Since, the  $dist(c_i,c_6)=1.78(<1.8)$ for all $i\in[5]$,  the regular unit $k$-gon  $\sigma_6$ intersects each of the regular unit $k$-gons  $\sigma_1, \sigma_2,\ldots, \sigma_5$.  Hence, the result follows.
\end{proof}

\noindent
\textbf{Theorem~\ref{thm:main_zeta}(e).}  \textit{The independent kissing number for the family of  congruent hypercubes in $\IR^d$ is at least~$2^{d+1}$, where $d\geq 2$ is an integer.}

\begin{figure}[!ht]
 \centering
\centering
\includegraphics[scale=0.75]{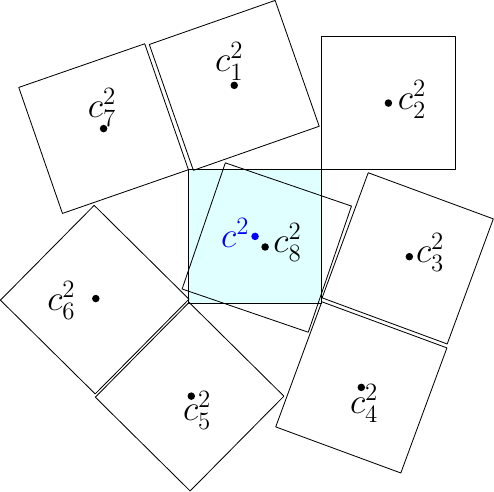}
\caption{An independent kissing configuration for congruent squares.}
\label{fig:arb_sq}
  \end{figure}

\begin{proof}
 To prove the theorem, we construct an  independent kissing configuration $K_d$, where the size of the independent set is $M(d)=8\times 2^{d-2}$ for unit hypercubes in $\IR^d$. 
 We give this construction by induction on $d$.
For $d=2$, the base case of the induction,  consider the independent kissing configuration given in Figure~\ref{fig:arb_sq}, where the independent set size is $8$.
 
Let us assume that the induction hypothesis holds for $d-1$ dimension. Let  $K_{d-1}$ be an independent kissing configuration, where the size of the independent set is $M(d-1)=8\times 2^{d-3}$ for unit hypercubes in $\IR^{d-1}$. Let the core hypercube in $K_{d-1}$ be centered at $c^{d-1}$, and let other hypercubes be centered at $c_1^{d-1}, c_2^{d-1},\ldots, c_{M(d-1)}^{d-1}$. To construct the independent kissing configuration $K_d$, we take two copies of the independent set from $K_{d-1}$ and append a $d$th coordinate to each. We append  $(\frac{1}{2}+\epsilon)$ as the   $d$th coordinate to all the centers of the first copy; whereas we append $-(\frac{1}{2}+\epsilon)$ to the other. Formally, we have

\begin{equation}\label{eq:mcd_arb_hypercube1}
c_i^d(x_j) =
    \begin{cases}
      c_i^{d-1}(x_j), & \text{ if $j<d$ and $i\leq M(d)$}\\
      \frac{1}{2}+\epsilon,& \text{ if $j=d$ and $i\leq M(d-1)$}\\
      -(\frac{1}{2}+\epsilon), & \text{ if $j=d$ and $i > M(d-1)$.}
    \end{cases}      
\end{equation}
The center $c^d$ is obtained by appending $0$ as the  $d$th coordinate to $c^{d-1}$. Therefore, we have
\begin{equation}\label{eq:mcd_arb_hypercube2}
c^d(x_j)=
     \begin{cases}
     $0$,& \text{ if $j=d$}\\
     c^{d-1}(x_j), & \text{ if $j<d$.}
     \end{cases} 
\end{equation}
Each of the  hypercubes is axis-parallel in every dimension except the first two coordinates. 
Thus, all unit hypercubes in $\IR^d$ centered at $c_1^d,c_2^d,\ldots,c_{M(d)}^d$ are mutually non-touching, and all of them are intersected by the unit hypercube centered at $c^{d}$. It is straightforward to see that the value of $M(d)= 2M(d-1)$, where $M(d-1)=8\times 2^{d-3}$. 
This completes the proof.
  \end{proof}


Before presenting the remaining proofs, we give definitions of ball packing and packing density that we are going to use.
A \emph{ball packing} is an arrangement of non-overlapping balls (of equal size) in $\IR^d$ within an object $O$ in $\IR^d$.  The aim of the packing problem is to find an arrangement in which the balls occupy as much space as possible within $O$. The proportion of space (volume) of $O$ occupied by the non-overlapping balls is called the \emph{packing density} $\eta$ of the arrangement.\\

\noindent
\textbf{Theorem~\ref{thm:main_zeta}(f).}  \textit{The independent kissing number for the family of bounded scaled disks in $\IR^2$ having radii in the interval $[1,2]$ is~11.}
\begin{proof}
    Note that the value of the independent kissing number for bounded scaled disks having radii in the interval $[1,2]$ is upper bounded by the number of unit disks that can be packed (non-overlapping unit disks) inside a disk of radius four. Melissen~\cite{melissen1994} proved that a disk with radius $1+\frac{1}{\sin{\frac{\pi}{9}}} \equiv3.923...$  can pack eleven unit disks inside it. On the other hand, Fodor~\cite{Fodor2000} proved that the smallest disk that can pack twelve unit disks inside it has a radius of at least~4.092.  Hence, the result follows. 
    \end{proof}


\noindent
\textbf{Theorem~\ref{thm:main_zeta}(g).}  \textit{The independent kissing number for the family of bounded-scaled $\alpha$-fat objects in $\IR^d$ having widths in the interval $[1,m]$ is at least~$\left(\frac{\alpha}{2}\left(\frac{m+2}{1+\epsilon}\right)\right)^d$ and at most~$\Big(\frac{m}{\alpha}+2\Big)^d$, where $\epsilon>0$ is a very small constant and $d\in\mathbb{Z}^{+}$.}

\begin{proof}



Note that an optimal independent kissing configuration is achieved when the core  $\sigma$ is the largest possible $\alpha$-fat object, i.e., with a width $m$, and all the mutually non-touching neighbours of $\sigma$ are the smallest $\alpha$-fat objects with unit width.
Now, we present an estimation on the upper and lower bound of $\zeta$ for the family of $\alpha$-fat objects having widths in $[1,m]$.

\noindent
   \textbf{Upper Bound.}    
  As per the definition of $\alpha$-fat object, a ball $B$ of radius $\frac{m}{\alpha}$ contains $\sigma$. On the other hand, each unit width $\alpha$-fat object contains a unit ball. To give an upper bound on the value of $\zeta$, we place as many non-touching unit balls as possible in and around the ball $B$ such that all unit balls intersect $B$. In other words, we pack (place non-overlapping objects) as many non-touching unit balls as possible within a ball $B_1$ of radius $\left(\frac{m}{\alpha}+2\right)$. 
  Therefore, the value of $\zeta$ is at most~$\Big(\frac{m}{\alpha}+2\Big)^d$, which is the ratio between the volume of the ball $B_1$ and the volume of a unit ball.

\noindent
    \textbf{Lower Bound.}
    As per the definition of $\alpha$-fat object, a ball $B'$ of radius $m$ is contained within $\sigma$. Each unit width $\alpha$-fat object is totally contained within a ball of radius $\frac{1}{\alpha}$. Thus, to obtain a lower bound on the value of independent kissing number $\zeta$, we place as many balls of radius $\frac{1}{\alpha}$ as possible in and around $B'$ such that all balls of radius $\frac{1}{\alpha}$ intersect $B'$. In other words, we pack (place non-overlapping objects) as many balls of radius $\Big(\frac{1+\epsilon}{\alpha}\Big)$ as possible within a ball $B_1'$ of radius $m+2$, where $\epsilon>0$ is a very small constant.
    Note that the value of $\zeta$ is at least the packing density times the ratio between the volume of the ball having radius $m+2$ and the volume of the ball having radius $\frac{(1+\epsilon)}{\alpha}$.
 Due to the Minkowski–Hlawka theorem, the packing density $\eta$ is at least $2^{-d}$~\cite{Sloane84}.
   Hence, the value of $\zeta$ is at least~$\left(\frac{\alpha}{2}\left(\frac{m+2}{1+\epsilon}\right)\right)^d$, which is the packing density (i.e., $2^{-d}$) times the ratio between the volume of the ball $B_1'$ and the volume of the ball having radius~$\frac{(1+\epsilon)}{\alpha}$.
\end{proof}

\begin{remark}\label{zeta_zeta'}
    Let $\zeta$ and $\zeta'$ be the independent kissing number for the family of bounded-scaled $\alpha$-fat objects in $\IR^d$ having widths in $[1,m]$ and $[1,2]$, respectively.  Now, we have the following remarks.
\begin{itemize}

\item[(i)] $\zeta'\leq\left(\frac{2}{\alpha}+2\right)^d$: It follows from Theorem~\ref{thm:main_zeta}(g) by putting the value of $m=2$.

 \item[(ii)]  $\zeta\geq \left(\frac{(m+2)\alpha^2}{4(1+\epsilon)(1+\alpha)}\right)^d \zeta'$:  Due to Theorem~\ref{thm:main_zeta}(g), we have   $\zeta\geq \left(\frac{\alpha}{2}\left(\frac{m+2}{1+\epsilon}\right)\right)^d$. Note that $\zeta\geq \left(\frac{\alpha}{2}\left(\frac{m+2}{1+\epsilon}\right)\right)^d
=\left(\frac{\alpha}{2}\left(\frac{m+2}{1+\epsilon}\right)\right)^d\left(\frac{1}{\frac{2}{\alpha}+2}\right)^d\left(\frac{2}{\alpha}+2\right)^d$. 
Now, using the fact $\zeta'\leq\left(\frac{2}{\alpha}+2\right)^d$ in the above expression, we have 
$\zeta\geq \left(\frac{(m+2)\alpha^2}{4(1+\epsilon)(1+\alpha)}\right)^d \zeta'$.
\end{itemize}
\end{remark}


\section{Applications}\label{Application_zeta}
In this section, we mention some of the implications of Theorem~\ref{thm:main_zeta}. 
Combining Theorem~\ref{thm:main_zeta} with Theorem~\ref{th:mds_geometric_objects_1}, Theorem~\ref{th:mids_geometric_objects} and Lemma~\ref{lem:coloring}, respectively, we obtain the following results for the online MDS, MIDS and MC problems, respectively.

\begin{theorem}\label{thm:main:conclusion1}

For each of the  MDS, MIDS and MC problems, there exists a deterministic online algorithm that achieves a competitive ratio of at most
 \begin{itemize}
 \item [(a)] 12 for  congruent balls in $\IR^3$;
     \item[(b)] $2^d$ for translated copies of a hypercube in $\IR^d$, where $d\in\mathbb{Z}^{+}$;
     \item[(c)]  $6$ for translated copies of a regular $k$-gon $($for $k=3$ and $k\geq 5)$;
      \item[(d)] $11$ for bounded scaled disks in $\IR^2$ having radii in the interval $[1,2]$; 
     \item[(e)] $\Big(\frac{m}{\alpha}+2\Big)^d$ for bounded-scaled $\alpha$-fat objects in $\IR^d$ having widths in the interval $[1,m]$, where $d\in\mathbb{Z}^{+}$.
    
     \end{itemize}
 \end{theorem}

For the MC problem, further, using Theorem~\ref{thm:AL} and Theorem~\ref{thm:main_zeta}(f,g), we have the following result.

\begin{theorem}\label{thm:main:conclusionMCAL}

For the MC problem, the  algorithm $\AL$  achieves a competitive ratio of at most
 \begin{itemize}
 \item [(a)] $11(\lfloor\log m \rfloor+1)$ for bounded scaled disks in $\IR^2$ having radii in the interval $[1,m]$; 
     \item[(b)] $\Big(\frac{2}{\alpha}+2\Big)^d (\lfloor\log m \rfloor+1)$ for bounded-scaled $\alpha$-fat objects in $\IR^d$ having widths in the interval $[1,m]$, where $d\in\mathbb{Z}^{+}$.
    
     \end{itemize}
 \end{theorem}
Here, we would like to point out that the result on Theorem~\ref{thm:main:conclusionMCAL}(a) improves the best-known competitive ratio from $28(\lfloor\log m \rfloor+1)$~\cite{ErlebachF06}.
Due to  Theorem~\ref{thm:AFF} and Theorem~\ref{thm:main_zeta}(g), we have the following.

 \begin{theorem}\label{thm:main:conclusionMCAFF}

For  the MC problem, the  algorithm $\AFF$  achieves a competitive ratio of
at most $O(\Big(\frac{2}{\alpha}+2\Big)^d \log m )$ for bounded-scaled $\alpha$-fat objects in $\IR^d$ having widths in the interval $[1,m]$, where $d\in\mathbb{Z}^{+}$.
    
 \end{theorem}

For the MCDS problem, combining Theorem~\ref{cds} and Theorem~\ref{thm:main_zeta}, we have the following.

 \begin{theorem}\label{thm:main:conclusion2}

For the  MCDS problem, there exists a deterministic online algorithm that achieves a competitive ratio of at most
 \begin{itemize}
 \item [(a)]  22 for  congruent balls in $\IR^3$;
     \item[(b)]  $2(2^d-1)$ for translated copies of a hypercube in $\IR^d$, where $d\in\mathbb{Z}^{+}$;
\item[(c)]  $10$ for translated copies of a regular $k$-gon $($for $k=3$ and $k\geq 5)$;
\item[(d)] ~$20$ for bounded scaled disks in $\IR^2$ radii in the interval $[1,2]$; 
 \item[(e)]  $2\left(\Big(\frac{m}{\alpha}+2\Big)^d-1\right)$ for bounded-scaled $\alpha$-fat objects having widths in the interval~$[1,m]$, where $d\in\mathbb{Z}^{+}$.
 
     \end{itemize}
 \end{theorem}

Now, we present the implication of Theorem~\ref{thm:main_zeta} on the online $t$-relaxed coloring problem. 
For a  nonnegative integer  $t$, 
in the \emph{ online $t$-relaxed coloring} problem, upon the arrival of a new vertex,
the algorithm must immediately assign a color to it, ensuring that the maximum degree of the subgraph induced by the vertices of this color class does not exceed~$t$. The objective of the problem is to minimize the number of distinct colors. Combining the result of Capponi and Pilotto~\cite[Thm~5]{CapponiP05} with Theorem~\ref{thm:main_zeta},  we have the following.

\begin{theorem}\label{thm:main:conclusion3}
For the online $t$-relaxed coloring problem, there exists an online algorithm that achieves an asymptotic competitive ratio of at most
 \begin{itemize}
 \item [(a)] 288 for  congruent balls in $\IR^3$;
     \item[(b)] $2^{2d+1}$ for translated copies of a hypercube in $\IR^d$, where $d\in\mathbb{Z}^{+}$;
     \item[(c)] at most $72$ for translated copies of a regular $k$-gon $($for $k=3$ and $k\geq 5)$;
     \item[(d)]  242 for bounded scaled disks in $\IR^2$ having radii in the interval $[1,2]$;
      \item[(e)]~$2\Big(\frac{m}{\alpha}+2\Big)^{2d}$ for bounded-scaled $\alpha$-fat objects having widths in the interval $[1,m]$, where $d\in\mathbb{Z}^{+}$.
       
     \end{itemize}
 \end{theorem}

\noindent


\section{Conclusion}\label{10}

{We conclude by mentioning some open problems.} The results obtained in this paper, as well as the results obtained in
~\cite{CapponiP05,MaratheBHRR95} for other graph problems, 
 are dependent on the independent kissing number $\zeta$. Consequently, the value of~$\zeta$  becomes an intriguing graph parameter to investigate. 
For congruent balls in $\IR^3$ and translates of a hypercube in $\IR^d$, we prove that the value of $\zeta$ is tight and equals 12 and $2^d$, respectively. In contrast, the value of $\zeta$ for translates of a regular $k$-gon (for $k=3$ and $k\geq 5$) is either 5 or 6.
We propose to settle the value $\zeta$ for this case as an open problem.
For congruent hypercubes in $\IR^d$, we prove that the value of $\zeta$  is at least $2^{d+1}$;  on the other hand, since congruent hypercubes are $\frac{1}{\sqrt{d}}$-fat objects, due to Theorem~\ref{thm:main_zeta}(g), it follows that $\zeta$ is at most $(2+\sqrt{d})^d$. Bridging this gap would be an open question. It would be of independent interest to see parametrized algorithms for graphs considering $\zeta$ as a parameter.

 \section*{Acknowledgement} 

{The authors would like to thank anonymous reviewers for bringing to their attention articles~\cite{ButenkoKU11,DuD15,MaratheBHRR95} that they were unaware of.}

\section*{Conflict of Interest}
The authors declare that there are no financial and non-financial competing interests that are relevant to the content of this article.

\bibliographystyle{plainurl}

\bibliography{reference}

\end{document}